\documentclass[a4paper,UKenglish]{lipics}

 \usepackage{stmaryrd}
\usepackage{microtype}
\usepackage{extarrows}
\usepackage{algorithm,algorithmic}
\usepackage{tikz}
\usetikzlibrary{decorations.pathmorphing} 
\usepackage{booktabs}

\newcommand{\comInAlign}[1]{\text{\emph{(#1)}}}

\newcommand{\act}[1]{\xlongrightarrow{#1}}          
\newcommand{\rate}{\alpha}
\newcommand{\FB}{FB}
\newcommand{\BB}{BB}
\newcommand{\OSB}{F}
\newcommand{\ESB}{B}

\newcommand{\abs}[1]{\left\lvert#1\right\rvert} 
\newcommand{\RE}{\mathbb{R}}  
\newcommand{\multiset}[1]{\{\!|#1|\!\}}
\newcommand{\M}{\mathcal{M}}  
\newcommand{\ms}{\mathcal{MS}}
\newcommand{\MS}{\ms}
\newcommand{\pa}{\mathcal{H}} 
\newcommand{\Y}{H}
\newcommand{\fr}{\mathbf{fr}}        
\newcommand{\R}{\mathcal{R}}  
\newcommand{\crr}{\mathbf{crr}}        
\newcommand{\gr}{\mathbf{pr}}        
\newcommand{\reduce}[2]{{{#1}^{#2}}}            
\newcommand{\canonical}[1]{{\mathcal{X}^{#1}}}    

\newcommand{\messageAboutProofs}{For the sake of readability all proofs are provided in a separate appendix.}
\newcommand{\messageAboutK}{All discussed $\kappa$-encodings are shown in Appendix A.4 and are available for download.}

\newtheorem{proposition}[theorem]{\textbf{Proposition}}

\title{Forward and Backward Bisimulations for Chemical Reaction Networks}
\titlerunning{Forward and Backward Bisimulations for Chemical Reaction Networks} 

\author[1]{Luca Cardelli}
\author[2]{Mirco Tribastone}
\author[3]{Max Tschaikowski}
\author[4]{\\ Andrea Vandin}
\affil[1]{Microsoft Research \& University of Oxford, UK\\
	\texttt{luca@microsoft.com}}
\affil[2-4]{IMT Institute for Advanced Studies Lucca, Italy\\
	\texttt{\{mirco.tribastone,max.tschaikowski,andrea.vandin\}@imtlucca.it}}
\authorrunning{L. Cardelli et al.} 

\Copyright{Luca Cardelli, Mirco Tribastone, Max Tschaikowski, and Andrea Vandin} 


\serieslogo{}

\allowdisplaybreaks[3]

\begin{document}
	
	\maketitle
	
	\begin{abstract}
    We present two quantitative behavioral equivalences over species of a chemical reaction network (CRN) with semantics based on ordinary differential equations. \emph{Forward CRN bisimulation} identifies a partition where each equivalence class represents the exact sum of the concentrations of the species belonging to that class. \emph{Backward CRN bisimulation} relates species that have identical solutions at all time points when starting from the same initial conditions. Both notions can be checked using only CRN syntactical information, i.e., by inspection of the set of reactions. We provide a unified algorithm that computes the coarsest refinement up to our bisimulations in polynomial time. Further, we give algorithms to compute quotient CRNs induced by a bisimulation. As an application, we find significant reductions in a number of models of biological processes from the literature. In two cases we allow the analysis of benchmark models which would be otherwise intractable due to their memory requirements.	
    \end{abstract}

	\section{Introduction}\label{sec:intro}
	
	At the interface between computer science and systems biology is the idea that biological systems can be interpreted as computational processes~\cite{export:68022,Cardelli:2012aa}, leading to a number of formal methods  applied to study biomolecular systems~\cite{Blinov22112004,Danos200469,Heath2008239}. In this context, chemical reaction networks (CRNs), a popular mathematical model of interaction in natural sciences, can also be seen as a kernel concurrent language for natural programming.
	
	In this paper we present, for the first time to our knowledge, quantitative bisimulation equivalences for CRNs with the well-known interpretation based on ordinary differential equations (ODEs). (To make the paper self-contained, all background is given in Section~\ref{sec:background}.) In this semantics, each species is associated with an ODE giving the deterministic evolution of its concentration starting from an initial condition. Our bisimulations are equivalences over species that induce a reduced CRN that exactly preserves the dynamics of the original one. This is an important goal, especially in order to cope with the potentially very large number of species and reactions in many biological networks~\cite{16430778,DBLP:conf/lics/DanosFFHK10}.
	
	We study two equivalences, developed in the Larsen-Skou style of probabilistic bisimulation~\cite{Larsen19911}, that are based on two distinct ideas of observable behavior. \emph{Forward CRN bisimulation} yields an aggregated ODE where the solution gives the exact sum of the concentrations of the species belonging to each equivalence class. In \emph{backward CRN bisimulation}, instead, equivalent species have the same solution \emph{at all time points}; in other words, backward CRN bisimulation relates species whose ODE solutions are equal whenever they start from identical initial conditions. The use of ``forward'' and ``backward'' has a long tradition in models of computation based on labelled transition systems~\cite{De-Nicola:1990aa}. In the case of quantitative variants, for instance those defined for process algebra with a continuous-time Markov chain (CTMC) semantics~\cite{hermanns:mtipp,pepa,bucholz94papm,Bernardo:2007ij}, forward bisimulations are equivalences that induce a CTMC aggregation in the sense of ordinary lumpability~\cite{BuchholzOrdinaryExact}, where the probability of an equivalence class is equal to the sum of the probabilities of the states belonging to that class. This is found by checking conditions on the \emph{outgoing transitions} of related states in the transition diagram. A backward bisimulation induces a CTMC aggregation in the sense of weak lumpability~\cite{Feret2012137}, where all states in the same equivalence class have a time-invariant conditional probability distribution; exact lumpability is a special case where the conditional probability distribution is uniform, in the sense that any two states of each equivalence class have the same probability at any point in time whenever they have the same initial probabilities. It is found by relating states according to conditions on  their \emph{predecessor states}~\cite{Feret2012137,1703385,BuchholzOrdinaryExact}.
	
	
	Despite being similar in spirit, technically our bisimulations are fundamentally different for two reasons. First, they concern a continuous-state semantics based on ODEs instead of a discrete-state CTMC. Second, they operate at the structural, syntactical level, because they are defined with quantities that can be computed by only inspecting the reactions of a CRN. Still, the repercussions of our bisimulations on the semantics are explained according to certain theories of aggregation. In particular, forward CRN bisimulation yields an aggregated system in the sense of ODE lumpability~\cite{LiRabitz1997,okino1998}. This theory covers linear transformations of the original state variables in general; here we consider an instance, which we call \emph{ordinary fluid lumpability}, where the transformation is induced by a partition of state variables. (Forward bisimulation is presented in Section~\ref{sec:ofsb}.) Backward bisimulation (presented in~Section~\ref{sec:efsb}) is related to \emph{exact fluid lumpability}, introduced in the context of process algebra with fluid semantics~\cite{concur2012}, identifying process terms with the same ODE solution when initialized equally.
	The disadvantage of forward CRN bisimulation is that it is lossy (yet exact) because, similarly to the forward stochastic analogues, from the aggregated ODE system in general it is not possible to recover the solutions for the individual species within the same equivalence class. On the other hand, it does not impose restrictions on the initial conditions, which instead are present in our backward variant. As a further important difference, forward CRN bisimulation (again, like its stochastic analogues) turns out to be a sufficient condition for ODE lumpability. Instead, backward CRN bisimulation enjoys a full characterization, in the sense that there exists a backward CRN bisimulation between two species if and only if they have the same ODE solutions (provided that they start from  equal initial conditions). More in general, by means of a number of examples we will show that the two equivalences are complementary because not comparable. In other words, there exist models that can be reduced up to forward CRN bisimulation but not by the backward variant, and vice versa; at the same time, there are models that can be reduced by both.
	
	To enhance the usefulness of these notions, we present (in Section~\ref{sec:pr}) a  \emph{template} partition-refinement algorithm that is parametric with respect to the  equivalence of interest, computing the coarsest refinement up to either variant in polynomial time. 
	To use our equivalences as an automatic model reduction tool, we further give two algorithms (in Section~\ref{sec:reducedCRN}) that provide the quotient CRN induced by either bisimulation.  With a prototype implementation available at \url{http://sysma.imtlucca.it/crnreducer/}, we show (in Section~\ref{sec:caseStudies}) that we are able to reduce  a number of case studies taken from the literature. Our bisimulations yielded quotient CRNs with number of reactions and species up to four orders of magnitude smaller than the original CRNs, leading to speed-ups in the ODE solution runtimes of up to five orders of magnitude. In two cases, it was possible to analyze models that were otherwise intractable directly within our experimental environment due to excessive memory requirements.

	\noindent {\bfseries Related work.}  Behavioral equivalences have been recently proposed in~\cite{Shin14} for comparing CRNs; however, the analysis is carried out at the qualitative level, i.e., ignoring the dynamical evolution.
	In~\cite{concur2012} is introduced the notion of \emph{label equivalence} for process algebra with fluid semantics, which captures exact fluid lumpability (processes are equivalent whenever their ODE solutions are equal at all time points). However, unlike backward CRN bisimulation, label equivalence is only a sufficient condition for ODE reduction. Indeed, it works at a coarser level of granularity as it relates \emph{sets} of ODE variables, each corresponding to the behavior of a sequential process. Instead, backward CRN bisimulation relates individual ODE variables. Further, the conditions for equivalence, specific to the process algebra, are difficult to check automatically because of the universal quantifiers over the ODE variables. More important, no algorithm for computing the coarsest partition was developed. Similar considerations apply to the process-algebra specific  ordinary fluid lumpability in~\cite{jlamp14}.
	
	Cardelli's notion of \emph{emulation} between two CRNs is a (structural) mapping of species and reactions that, like backward CRN bisimulation,  guarantees the equality between the ODE solutions at all time points~\cite{CardelliCRN}. An emulation requires a source and a target CRN --- the modeler is intended to have the suspicion that, for some given CRN, another CRN might be related to it. But emulation cannot be used when one wants to discover equivalences between species \emph{within the same given CRN}. Thus, emulation is not useful for model reduction because a-priori information about the structure of a quotient CRN is not available. Furthermore, no algorithm is provided in~\cite{CardelliCRN} to find emulations automatically.   Since backward CRN bisimulation fully characterizes exact fluid lumpability, it is not difficult to show that backward CRN bisimulation generalizes emulation in the sense that any emulation between two CRNs can be understood in terms of a backward CRN bisimulation over the species of a ``union CRN'' that contains all the reactions of the two CRNs of interest.
	
	
	Model reductions have been studied in related models for biomolecular networks (e.g.~\cite{DBLP:conf/lics/DanosFFHK10,Feret_IJSI2013,DBLP:journals/entcs/CamporesiFKP10}), most notably for rule-based systems such as BioNetGen~\cite{Blinov22112004} and the $\kappa$ calculus~\cite{Danos200469}. These offer an intensional modeling approach, by providing graph-rewrite rules of interaction instead of a complete enumeration of all chemical reactions involved.
	\emph{Differential fragments} for $\kappa$ are self-consistent aggregates found by a static analysis on the model,  identifying sums of chemical species for which an ODE system can be explicitly written~\cite{DBLP:conf/lics/DanosFFHK10}. In this sense, they are analogous to our CRN bisimulations, but with notable differences. First, fragmentation works directly at the rule-based level. This has the advantage that the analysis is performed on a set of rewrite rules, which is typically much more compact than the fully enumerated CRN. However, fragmentation is domain-specific, hence the model must be conveniently expressed as a biomolecular system (e.g., with complex formation or internal state modification). On the other hand, CRN bisimulations work for a generic language-independent CRN, which however must be explicitly given. Further, unlike CRN bisimulations, fragmentation is performed on  a ``static'' view of the model, without information on the reaction rates.
	The ODE aggregations of both forward CRN bisimulation and fragmentation introduce loss of information (in contrast to backward CRN bisimulation).
	But, unlike our forward variant, in fragmentation the same species may be present in more than one fragment.
	Additionally, species may occur in fragments with multiplicity numbers.
	Thus, fragmentation can be seen as a form of \emph{improper lumping} that is not necessarily induced by a partition of the original state-space variables~\cite{okino1998}.
	%
	%
	Overall, because of these differences, it is not difficult to find models that can be reduced by our CRN bisimulations and not by fragmentation, and vice versa.
	This is presented in detail in Section~\ref{sec:caseStudies}.

	\section{Background}\label{sec:background}
	\subparagraph*{Notation.}
	We write $A \to B$ and $B^A$ for the functions from $A$ to $B$. When $f \in A \to B$ and $a \in A$, we set $f_a := f(a)$. Moreover, for any $X \subseteq A$ and $b \in B$, we define $f(X) := \{ b \in B \mid \exists a \in X . (f(a) = b) \}$. 
	Sets and multisets are denoted by $\{ \ldots \}$ and $\{| \ldots |\}$, respectively. Also, we shall not distinguish among an equivalence relation and the partition induced by it, and shall use the symbol $\sim_\pa$ to denote the equivalence relation with $\pa = S \!/\!\!\sim_\pa$.
	Finally, given two partitions $\mathcal{H}_1$ and $\mathcal{H}_2$ of a given set $S$, we say that $\mathcal{H}_1$ is a \emph{refinement} of $\mathcal{H}_2$ if for any $H_1 \in \mathcal{H}_1$ there exists a (unique) $H_2 \in \mathcal{H}_2$ such that $H_1 \subseteq H_2$.

	\subsection{Chemical Reaction Networks}\label{sec:crn}
	Formally, a CRN $(S, R)$ is a pair consisting of a finite set of species $S$ taken from a countable infinite universe of all species,
	and a finite set of chemical reactions $R$. A reaction is a triple written in the form $\rho \act{\rate} \pi$, where $\rho$ and $\pi$ are the multisets of species 
	\emph{reactants} and \emph{products}, respectively, and $\rate > 0$ is the reaction rate.
	In particular, we focus on basic chemistry where only \emph{elementary reactions} are considered, where at most two reactants (possibly of the same species) interact. No restrictions are instead imposed on products. Several models found in the literature (including those discussed in Section~\ref{sec:caseStudies}) belong to this class. Also, this is consistent with the physical considerations which stipulate that reactions with more than two reactants are very unlikely to occur in nature~\cite{citeulike:1079741}.
	We denote by $\rho(X)$ the multiplicity of species $X$ in the multiset $\rho$, and by $\ms(S)$ the set of finite multisets of species in $S$. To adhere to standard chemical notation, we shall use the operator $+$ to denote multiset union, e.g., $X + Y + Y$ (or just $X + 2Y$) denotes the multiset $\multiset{X,Y,Y}$.
	We may also use $X$ to denote either the species $X$ or the singleton $\multiset{X}$.

	The (autonomous) ODE system $\dot{V} = F(V)$ underlying a CRN $(S,R)$ is 
	$F : \mathbb{R}^S_{\geq0} \to \mathbb{R}^S$, where each component $F_X$, with $X\in S$ is defined as:
	\[F_X(V) := \sum_{\rho \act{\alpha}\pi \in R} (\pi(X)-\rho(X))\cdot\alpha\cdot\prod_{Y \in S} V_Y^{\rho(Y)}\ .\]
	This represents the well-known \emph{mass-action} kinetics, where the reaction rate is proportional to the concentrations of the reactants involved.
	Since the ODE system of a CRN is given by polynomials, the vector field $F$ is locally Lipschitz. Hence, the theorem of Picard-Lindel\"{o}f ensures that for any $V(0) \in \RE_{\geq0}^S$ there exists a unique non-continuable solution of $\dot{V} = F(V)$.
	
	\begin{example}\label{ex:model}
		We now provide a simple CRN, $(S_e, R_e)$, with $S_e = \{ A,B,C,D,E\}$ and
		$R_e \! = \! \{ A\!\act{6}\!E, B\!\act{6}\!D, {A\!+\!B}\!\act{2}\!{C}, {C\!+\!D}\!\act{5}\!{2C\!+\!D}, {E\!+\!D}\!\act{5}\!{2E\!+\!D} \}$,
		which will be used as a running example throughout the paper. Its ODE system is given by
		\begin{align}
			\dot{V}_{A} & = - 6 \, V_{A} - 2 \, V_{A}\,V_{B}  & \dot{V}_{B}  & = - 6 \, V_{B} - 2 \, V_{A}\,V_{B} &  \dot{V}_{C} & = 2 \, V_{A} \, V_{B} + 5 \, V_{C} \, V_{D}
			\nonumber \\
			\dot{V}_{D} & =  6 \, V_{B}  & \dot{V}_{E}  & = 6 \, V_{A}  + 5\, V_{E} \, V_{D}
			\nonumber
		\end{align}
	\end{example}
	
	In the following, we shall assume that the
	universe of all species is well-ordered with respect to $\sqsubseteq$. We then say that a function $\mu : S \to S$ is a \emph{choice function} of a partition $\pa$ of $S$, if $\mu(X) = \min_\sqsubseteq H$ for all $H \in \mathcal{H}$ and $X \in H$. Also, choice functions can be trivially lifted to multisets applying them element-wise, e.g., $\mu(X+Y) = \mu(X)+\mu(Y)$.
	
	\subsection{Fluid Lumpability}
	
	\subparagraph*{Ordinary Fluid Lumpability.}
	We start by defining the notion of ordinary fluid lumpability, which is an instance of \emph{ordinary lumpability} for ODEs~\cite{LiRabitz1997} specialized to CRNs.

	
	\begin{definition}[Ordinary fluid lumpability]\label{def:o-lumpability}
		Let $(S,R)$ be a CRN, $F$ be its vector field, and $\pa = \{H_1,\ldots,H_m\}$ a partition of $S$. Then, $\mathcal{H}$ is \emph{ordinary fluid lumpable} if for all $H \in \pa$ there exists a polynomial $\wp_H$ in $|\pa|$ variables such that $\sum_{X \in H} F_X(V) = \wp_H(\sum_{X \in H_1} V_X, \ldots, \sum_{X \in H_m} V_X)$ for all $V \in \RE_{\geq0}^S$.
	\end{definition}
	Informally, a partition $\pa$ is ordinary fluid lumpable if, for each $H \in \pa$, the polynomial $\sum_{X \in H} F_X(V)$ in the variables $\{V_X \mid X \in S\}$ can be rewritten into a polynomial $\wp_H$ in the variables $\{ \sum_{X \in H} V_X \mid H \in \pa\}$. In particular, if $\pa$ is known to be an ordinary fluid lumpable partition of $(S,R)$ and $V$ denotes the solution of $\dot{V} = F(V)$ subject to $V(0) \in \RE_{\geq0}^S$, the solution of the aggregated ODE system $(\dot{W}_{H_1}, \ldots, \dot{W}_{H_m}) = (\wp_{H_1}(W),\ldots,\wp_{H_m}(W))$ with $W_H(0) = \sum_{X \in H} V_X(0)$ is such that $W_H(t) = \sum_{X \in H} V_X(t)$ for all $t \in \text{domain}(V)$.
	
	\begin{example}\label{ex:OlumpedODEs}
		Consider the ODEs of $(S_e,R_e)$ of Example~\ref{ex:model}, and let 
		$\pa_{O}=\{\{A\},\{B\},\{C,E\},$ $\{D\}\}$.
		By applying a variable renaming consistent with the blocks of $\pa_O$, i.e., $V_{CE}=V_C+V_E$, and by exploiting the linearity of the differential operator we get
		\begin{align*}\!\!\!
			\dot{V}_{A} \!=\! - 6 V_{A} \!-\! 2 V_{A}V_{B}  \quad \
			\dot{V}_{B} \! = \!- 6 V_{B} \!-\! 2 V_{A}V_{B} \quad \
			\dot{V}_{CE} \!= \!2V_{A}V_{B} \!+\! 6V_{A} \!+\! 5V_{D}V_{CE} \quad \
			\dot{V}_{D} \!= \! 6V_{B}
		\end{align*}
		That is, we obtained an ODE system  in terms of block variables only. \qed
	\end{example}

	
	\subparagraph*{Exact Fluid Lumpability.}
	\!We extend to CRNs the notion of exact fluid lumpability in~\cite{concur2012}.
	
	\begin{definition}[Exact fluid lumpability]\label{def:e-lumpability}
		Let $(S,R)$ be a CRN, $F$ its vector field, and $\mathcal{H}$ a partition of $S$. We call $V \!\in\! \mathbb{R}^{S}$ \emph{constant on} $\mathcal{H}$ if $V_{X_i} = V_{X_j}$ for all $H \in \mathcal{H}$, and all $X_i,X_j \in H$.
		Then, $\mathcal{H}$ is \emph{exactly fluid lumpable} if $F(V)$ is constant on $\mathcal{H}$ whenever $V$ is constant on $\mathcal{H}$.
	\end{definition}
	\begin{example}\label{ex:ElumpedODEs}
		Consider the ODEs of $(S_e,R_e)$ of Example~\ref{ex:model},
		and let 
		$\pa_{E}=\{\{A,B\},\{C\},\{D\},$ $\{E\}\}$.
		It is easy to see that $A$ and $B$ 
		have same concentrations at all time points if initialized equally.
		In these cases, we can replace the ODEs of $(S_e,R_e)$
		with the ones aggregated according to $\pa_{E}$, obtained by removing $\dot{V}_{B}$ and replacing all occurrences of $V_{B}$ with $V_{A}$:
		\begin{align*}
			\dot{V}_{A} = - 6 \, V_{A} - 2 \, V_{A}\,V_{A}    \quad \
			\dot{V}_{C} = 2 \, V_{A} \, V_{A} + 5 \, V_{C} \, V_{D}   \quad \
			\dot{V}_{D} =  6 \, V_{A}    \quad \
			\dot{V}_{E} = 6 \, V_{A}  + 5\, V_{E} \, V_{D}
		\end{align*}
		That is, we obtained a (lossless) aggregated ODE system written	in terms of a variable per block, chosen according to $\sqsubseteq$.
		%
		\qed
	\end{example}
	

	
	We remark that the above definition expresses exact fluid lumpability in terms of properties of the ODE vector field of a CRN. Instead,  in~\cite{concur2012}  exactly fluid lumpability was defined directly in terms of the desired dynamical property, i.e., that the ODE solutions within any equivalence class be equal at all time points. The following result is a new contribution showing that this dynamical property is fully characterized by the vector-field based definition.
	
	\begin{theorem}\label{thm_efl}
		Let $(S,R)$ be a CRN and $F$ its vector field. A partition $\mathcal{H}$ of $S$ is exactly fluid lumpable if and only if, for any $V(0) \in \RE_{\geq0}^S$ that is constant on $\pa$, the underlying solution of $\dot{V} = F(V)$ is such that $V(t)$ is constant on $\pa$ for all $t \in \text{domain}(V)$.
		\footnote{
			\messageAboutProofs
		}
	\end{theorem}
	
	\section{CRN Bisimulations}\label{sec:sb}
	
	Both notions of fluid lumpability given in 
	Section~\ref{sec:background} are not convenient to be used directly because they involve a universal quantifier over the  (uncountable) state space.  We address this problem by providing structural conditions that concern only the reactions of a CRN. 
	

	\subsection{Forward CRN Bisimulation}\label{sec:ofsb}

	We now introduce forward CRN bisimulation, an equivalence on species that will turn out to induce ordinary fluid lumpability.
	We start with the notions of \emph{reaction} and \emph{production rate}.
	The former collects the rates at which the concentration of a species  $X$ decreases when reacting with a given partner.
	The latter collects the positive contribution that $X$ exerts to the concentration of a species $Y$, again when reacting with a certain partner.
	%
	
	
	\begin{definition}[Reaction and production rates]\label{def:conditionalReactionRate}
		Let $(S,R)$ be a CRN, $X,Y \in S$, and $\rho\in \ms(S)$.
		The \emph{$\rho$-reaction rate} of $X$, and the \emph{$\rho$-production rate} of $Y$\!\!-elements by $X$ are defined respectively as
		\begin{align*}
		\crr[X,\rho] & := (\rho(X)+1) \!\!\!\!\!\!  \sum_{X + \rho \act{\rate} \pi \in R} \!\!\!\!\!\! \rate , &
		\gr(X,\rho,Y) & := (\rho(X)+1) \!\!\!\!\!\! \sum_{X + \rho \act{\rate} \pi \in R} \!\!\!\!\!\! \rate \cdot \pi(Y)
		\end{align*}
		\noindent Finally, for $\Y \subseteq S$ we define  $\gr[X,\rho,\Y] := \sum_{Y\in \Y} \gr(X,\rho,Y) .$
	\end{definition}
	
	\begin{definition}[Forward CRN Bisimulation]\label{def:dsb}
		Let $(S,R)$ be a CRN, $\R$ an equivalence relation over $S$ and $\pa = S / \R$. Then, $\R$ is a forward CRN bisimulation (abbreviated \FB{}) if for all $(X,Y) \in \R$, all $\rho \in \MS(S)$, and all $\Y \in \pa$ it holds that
		\begin{equation}\label{eq_i_and_ii}
		\crr[X,\rho]= \crr[Y,\rho] \quad \text{and} \quad
		\gr[X,\rho,\Y] = \gr[Y,\rho,\Y]
		\end{equation}
		%
	\end{definition}
	
	
	
	\begin{example}\label{ex:dsb}
		Consider $\pa_{O}=\{\{A\},\{B\},\{C,E\},\{D\}\}$ of Example~\ref{ex:OlumpedODEs}.
		It can be shown that $\pa_{O}$ is an \FB{}, as, e.g., $\crr[C,\!D]\!=\!\crr[E,\!D]\!=\!5$, and $\gr[C,\!D,\{C\!,\!E\}]\!=\!\gr[E,\!D,\{C\!,\!E\}]\!=\!10$.
	\end{example}

	We are interested in the coarsest \FB{}, as well as in the coarsest one refining a given initial partition of species. 

	\begin{proposition}\label{prop:cofsb:ex}
		Let $(S,R)$ be a CRN, $I$ a set of indices, and $\R_i$ an \FB{} for $(S,R)$, for all $i \in I$. The transitive closure of their union $\R \!=\! (\bigcup_{i \in I}\R_i)^*$ is an \FB{} for $(S,R)$. In particular, if each $R_i$ is such that $S / \R_i$ refines some partition $\mathcal{G}$ of $S$, then so does $S / \R$.
	\end{proposition}
	
	
	\begin{theorem}[Forward bisimulation implies ordinary fluid lumpability]\label{th:dsbAndODELumping}
		Let $(S,R)$ be a CRN. Then, $\pa$ is an ordinarily fluid lumpable partition of $S$ if $\pa$ is an \FB{} of $S$.
	\end{theorem}
	\FB{} is only a sufficient condition for lumpability, as discussed in the next example. (However, Section~\ref{sec:caseStudies} shows that \FB{} can be effectively applied to interesting existing models.)
	
	\begin{example}\label{ex:sufficientForLumpingODEs}
		Consider the CRN $(\{F,G\},\{F\act{\rate_1}G,G\act{\rate_2}F\})$, having ODEs
		\begin{align*}
		\dot{V}_{F}=-\rate_1\,V_F+\rate_2\,V_G\qquad\dot{V}_{G}=-\rate_2\,V_G+\rate_1\,V_F
		\end{align*}
		%
		If $\alpha_1\not=\alpha_2$, $\pa_c=\{\{G,F\}\}$ is not an \FB{}, as $\crr[F,\emptyset]=\alpha_1$ and $\crr[G,\emptyset]=\alpha_2$. Nevertheless, the above ODE system is lumpable. Indeed, by applying the variable renaming consistent with $\pa_c$, i.e., $V_{FG}=V_F+V_G$, we get a single ODE for $V_{FG}$, i.e., $\dot{V}_{FG}=0$.
		\qed
	\end{example}

	\subsection{Backward CRN Bisimulation}\label{sec:efsb}
	
	We now introduce backward CRN bisimulation, an equivalence on species that will turn out to characterize exact fluid lumpability.
	We start with the notion of cumulative flux rate, which collects the overall contribution that reactions with a given multiset of reactants $\rho$ exert to the concentration of a species $X$.
	
	\begin{definition}[Cumulative flux rate]\label{def:cumulativeFluxRate}
		Let $(S,R)$ be a CRN, $X \in S$, $\rho \in \ms(S)$, and $\M\subseteq\ms(S)$. Then, we define 
		\begin{align*}
		\fr(X,\rho) & := \sum_{\rho \act{\alpha} \pi \in R} (\pi(X)-\rho(X)) \cdot \alpha,  &
		\fr[X,\M] & := \sum_{\rho\in\M}\fr(X,\rho) .
		\end{align*}
		We call $\fr(X,\rho)$ and $\fr[X,\M]$ \emph{$\rho$-flux rate} and \emph{cumulative $\M$-flux rate} of $X$, respectively.
	\end{definition}
	
	
	%

	\begin{definition}[Backward CRN bisimulation]\label{def:efsb}
		Let $(S,R)$ be a CRN, $\R$ an equivalence relation over $S$, $\pa = S / \R$ and $\mu$ the choice function of $\pa$. Then, $\R$ is a backward CRN bisimulation (\BB{}) if for any $(X,Y) \in \R$ it holds that
		\begin{equation}\label{eq_iii}
		\fr[X,\M] = \fr[Y,\M] \ \ \text{for all } \ \ \M \in \{ \rho \mid \rho \act{\alpha} \pi \in R \} / \approx_\pa ,
		\end{equation}
		where any two $\rho,\sigma \in \ms(S)$ satisfy $\rho \approx_{\pa} \sigma$ if $\mu(\rho) = \mu(\sigma)$.
	\end{definition}
	
	
	\begin{example}\label{ex:efsb}
		Consider $\pa_{E}=\{\{A,B\},\{C\},\{D\},\{E\}\}$ of Example~\ref{ex:ElumpedODEs}.
		We first  note that $\multiset{A} \approx_{\pa_{E}} \multiset{B}$, as $\approx_{\pa_{E}}$ relates multisets 
		with same number of $\pa_{E}$-equivalent species.
		Also, it can be shown that $\pa_{E}$ is a \BB{}, as, e.g., $\fr[A,\M]=\fr[B,\M]=-6$ for $\M=\{\multiset{A},\multiset{B}\}$. 
		%
		\qed
	\end{example}
	
	As for \FB{}, there exists a coarsest \BB{} (that refines a given partition of $S$).
	
	\begin{proposition}\label{prop:cefsb:ex}
		Let $(S,R)$ be a CRN, $I$ a set of indices, and $\R_i$ a \BB{} for $(S,R)$, for all $i \in I$. The transitive closure of their union $\R \!=\! (\bigcup_{i \in I}\R_i)^*$ is a \BB{} for $(S,R)$. In particular, if each $R_i$ is such that $S / \R_i$ refines some partition $\mathcal{G}$ of $S$, then so does $S / \R$.
	\end{proposition}
	
	We now state the mentioned characterization of exact fluid lumpability in terms of \BB{}.
	
	\begin{theorem}[Backward bisimulation characterizes exact fluid lumpability]\label{th:efsbAndODELumping}
		Let $(S,R)$ be a CRN. Then, $\pa$ is an exactly fluid lumpable partition of $S$ if and only if $\pa$ is a \BB{} of $S$.
	\end{theorem}

	\begin{remark}
		We wish to stress that \FB{} and \BB{} are not comparable: First, $\!\pa_O\!$ is not a \BB{}, as $\fr[C,\!\{A\!+\!B\}]\!=\!2$ and $\fr[E,\!\{A\!+\!B\}]\!=\!0$; Second, $\!\pa_E\!$ is not an \FB{}, as $\crr(A,B)\!=\!2$ and $\crr(B,B)\!=\!0$; Third, for the same reasons, $\{\{A,\!B\},\{C,\!E\},\{D\}\}$ is neither an \FB{} nor a \BB{}.
		Similar examples on models of biological relevance are provided in Section~\ref{sec:caseStudies}.
		\qed
	\end{remark}
	
\section{Reduced Chemical Reaction Networks up to CRN Bisimulations}\label{sec:reducedCRN}

We have shown that, given a CRN and a CRN bisimulation $\R$, we can analyze the aggregated ODE system according to $\R$. We now provide the notion of reduced CRN from which the aggregated ODEs can be directly generated, as depicted in Figure~\ref{fig:CommutingDiagram}. 

\begin{figure}[t]
	\center
	\begin{tikzpicture}[every node/.style={midway}]
	\matrix[column sep={4em,between origins},
	row sep={1em}] at (0,0)
	{
		\node(CRN)   {CRN}  ; & & &
		\node(reducedCRN){reduced CRN};
		\\\\
		\node[text centered](semantics){ODEs}; & &&
		\node[text centered](lumpedSemantics){lumped ODEs};
		\\};
	\draw[->] (CRN)		-- (semantics)		node[anchor=east]  {\emph{semantics}};
	\draw [->,line join=round,
	decorate, decoration={
		zigzag,
		segment length=4,
		amplitude=.9,post=lineto,
		post length=2pt
	}]
	(CRN)		-- (reducedCRN)		node[anchor=south] {\emph{reduce wrt $\pa$}};
	\draw[->,dotted] (semantics)	-- (lumpedSemantics)	node[anchor=north]  {\emph{lump wrt $\pa$}};
	\draw[->] (reducedCRN)	-- (lumpedSemantics)	node[anchor=west]  {\emph{semantics}};
	\end{tikzpicture}
	\caption{Relation among ($\pa$-reduced) CRNs and ($\pa$-lumped) semantics, with $\pa$ a bisimulation.}
	\label{fig:CommutingDiagram}
\end{figure}
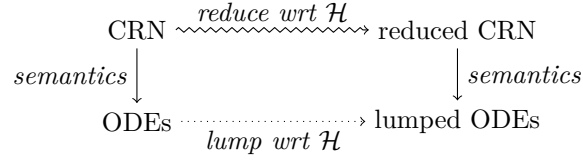


\begin{definition}[Forward reduction]\label{def:OReduction}
	Let $(S,R)$ be a CRN, $\pa$ an \FB{}, and $\mu$ its choice function. The \emph{${(\pa,F)}$-reduction of $(S,R)$} is given by $\reduce{(S,R)}{{(\pa,F)}} = (\reduce{S}{{(\pa,F)}},\reduce{R}{{(\pa,F)}})$, where $\reduce{S}{{(\pa,F)}} = \mu(S)$ and $\reduce{R}{{(\pa,F)}}$ is defined as follows: (F1) Discard all reactions $\rho \act{\alpha} \pi$ such that $\rho \not=\mu(\rho)$; (F2) Replace all remaining reactions $\rho \act{\alpha} \pi$ with $\rho \act{\alpha} \mu(\pi)$; (F3) Fuse all reactions that have the same reactants and products by summing their rates.
\end{definition}

The idea underlying forward reduction is to discard all reactions having non-representative reagents, and to replace the products of the remaining reactions with their representatives. This 
can be seen as a special case of Theorem~4.4 of~\cite{DBLP:journals/entcs/CamporesiFKP10}. 

\begin{example}\label{ex:reducedModel}
	Consider the \FB{} $\pa_{O}=\{\{A\},\{B\},\{C,E\},\{D\}\}$ used in Example~\ref{ex:OlumpedODEs}. The $(\pa_O,F)$-reduction of $(S_e, R_e)$ is (with $C$ being the representative of its block) ${S_e}^{\!\!\!(\pa_O,F)} \! = \! \{ A, B, C, D\}$, ${R_e}^{\!\!\!(\pa_O,F)} \! = \! \{ A\!\act{6}\!C,  B\!\act{6}\!D, A\!\!+\!\!B\!\act{2}\!C, {C\!\!+\!\!D}\!\act{5}\!{2C\!\!+\!\!D}\}$. Note that the reaction ${E\!\!+\!\!D}\!\act{5}\!{2E\!\!+\!\!D}$ is discarded, as $E$ is not a representative species.\qed
\end{example}

%

We now state that the ${(\pa,F)}$-reduction of an \FB{} 
$\pa$ induces the ODEs aggregated according to $\pa$. For example, the ${(\pa_O,F)}$-reduction of $(S_e,R_e)$ induces the ODEs shown in Example~\ref{ex:OlumpedODEs}, if applying the renaming $V_C=V_{CE}$.

\begin{theorem}[Forward reduction induces aggregation]\label{th:reducedModelHasLumpedODEs}
	Let $(S,R)$ be a CRN, $\pa$ an \FB{} and $\mu$ its choice function. Then, $\reduce{(S,R)}{{(\pa,F)}}$ is computed in at most $\mathcal{O} \big( |R| \cdot |S| \cdot (\log (|R|) +\log (|S|)) \big)$ steps. Crucially, if $F$ is the vector field of $(S,R)$ and $\hat{F}$ the one of $\reduce{(S,R)}{{(\pa,F)}}$, then $\sum_{X\in\Y} F_X(V) = \hat{F}_{\mu(Y)}(\sum_{X \in H_1} V_X, \ldots, \sum_{X \in H_m} V_X)$ for all $V \in \RE_{\geq0}^S$, $\Y\in\pa$ and $Y \in \Y$.
\end{theorem}


For the backward reduction, the underlying idea is to keep track only of differential contributions that affect the representative species $\mu(S)$. 

\begin{definition}[Backward reduction]\label{def:EReduction}
	Let $(S,R)$ be a CRN, $\pa$ a \BB{}, and $\mu$ its choice function. The \emph{${(\pa,B)}$-reduction of $(S,R)$} is given by $\reduce{(S,R)}{{(\pa,B)}} = (\reduce{S}{{(\pa,B)}},\reduce{R}{{(\pa,B)}})$, where $\reduce{S}{{(\pa,B)}} = \mu(S)$ and $\reduce{R}{{(\pa,B)}}$ is obtained as follows: (B1) Replace all reactions $\rho \act{\alpha} \pi$ with $\rho \act{\alpha} \tilde{\pi}$ where $\tilde{\pi}(X_i) := \pi(X_i)$ if $X_i \in \mu(S)$ and $\tilde{\pi}(X_i) := \rho(X_i)$ otherwise; (B2) Replace all such obtained reactions $\rho \act{\alpha} \pi$ with $\mu(\rho) \act{\alpha} \mu(\pi)$; (B3) Fuse all reactions that have the same reactants and products by summing their rates.
\end{definition}

\begin{example}
	Considering the CRN $(S_e,R_e)$ and the \BB{} $\pa_E$, (B1) changes $B\!\!\act{6}\!\!D$ in $\!B\!\!\act{6}\!\!D\!+\!B$, and ${A\!+\!B}\!\act{2}\!{C}$ in ${A\!+\!B}\!\act{2}\!{C\!+\!\!B}$, while \!(B2)\! yields $\!\{\!A\!\!\act{6}\!\!E, A\!\!\act{6}\!\!D\!+\!A, {A\!+\!A}\!\!\act{2}\!\!{C\!+\!A}, {C\!+\!D}\!\!\act{5}\!\!{2C\!+\!D}, {E\!+\!D}\!\!\act{5}\!\!{2E\!+\!D} \}$. Finally, (B3) does not introduce any change. \qed
\end{example}

\begin{theorem}[Backward reduction induces aggregation]\label{th:e:reduction}
	Let $(S,R)$ be a CRN, $\pa$ a \BB{} and $\mu$ its choice function. Then, $\reduce{(S,R)}{{(\pa,B)}}$ is computed in at most $\mathcal{O}\big(|R| \cdot |S| \cdot (\log(|R|) + \log(|S|))\big)$ steps. Crucially, if $\hat{F}$ denotes the vector field induced by $\reduce{(S,R)}{{(\pa,B)}}$, it holds that $F_X(V) = \hat{F}_X(V)$ for all $X \in \mu(S)$ and $V \in \RE_{\geq0}^S$ that are constant on $\pa$.
\end{theorem}
	
	
	\section{Partition Refinement Algorithms for CRN Bisimulations}\label{sec:pr}
	We study
	a polynomial-time algorithm for the computation of the coarsest bisimulations that refine an arbitrary input partition. 
	We start introducing two auxiliary equivalence relations.

	\begin{definition}[Splitter equivalences]
		Let $(S,R)$ be a CRN and $\pa$ a partition over $S$. Then, we write $X \sim_\pa^\OSB{} Y$ if (\ref{eq_i_and_ii}) is fulfilled by $(X,Y)$. Similarly, write $X \sim_{\mathcal H}^\ESB{} Y$ if $(X,Y)$ satisfies (\ref{eq_iii}).
	\end{definition}
	
	Algorithm~\ref{algorithm_cfsb} iteratively computes the coarsest forward or backward bisimulation (when $\chi = F$ or $\chi = B$, respectively) that refines a given input partition of species of a CRN.
	Note that, contrary to CRN reduction algorithms, one (parametric) algorithm suffices for both bisimulations.
	%
	Using the above splitter equivalences, at each iteration the blocks of the current partition $S/\!\!\sim_\pa$ are split in sub-blocks of $\sim^\chi_\mathcal{H}$-equivalent species $S / (\sim^\chi_\mathcal{H} \cap \sim_\mathcal{H})$. The algorithm terminates when no refinement is performed.
	
	The freedom in choosing the initial partition $\mathcal{G}$ is useful in both bisimulations. For \FB{} it allows to single out species that are the ``observables'' of the CRN.
	These are the species for which the modeler is interested in obtaining distinct ODE solutions, information which would otherwise be lost if such species are found in larger equivalence classes. \BB{} is lossless, hence this issue does not arise. However \BB{} requires the same initial conditions for equivalent species. In this case, an appropriate input partition may tell apart species for which it is known that the initial conditions are different.
	
	\begin{algorithm}[t!]
		\caption{Template partition refinement algorithm for the construction of the coarsest CRN bisimulations that refine some given initial partition $\mathcal{G}$.}\label{algorithm_cfsb}
		\begin{algorithmic}
			\REQUIRE A CRN $(S,R)$, a partition $\mathcal{G}$ of $S$ and $\chi \in \{F,B\}$.
			
			\STATE $\mathcal{H}$ $\longleftarrow$ $\mathcal{G}$
			\WHILE{\TRUE}
			\STATE $\mathcal{H}' \longleftarrow$ $S / (\sim^\chi_\mathcal{H} \cap \sim_\mathcal{H})$
			
			\IF{$\mathcal{H}' = \mathcal{H}$}
			\RETURN $\mathcal{H}$
			\ELSE
			\STATE $\mathcal{H} \longleftarrow \mathcal{H}'$
			\ENDIF
			\ENDWHILE
		\end{algorithmic}
	\end{algorithm}
	
	\begin{theorem}[Correctness]\label{thm_cfsb}
		Given a CRN $(S,R)$ and a partition $\mathcal{G}$ of $S$, Algorithm~\ref{algorithm_cfsb} calculates the coarsest forward and backward bisimulation that refines $\mathcal{G}$.
		In both cases, the number of steps needed 
		is polynomial in the number of species and reactions.
	\end{theorem}
	
	Note that, due to space constraints, we only focussed on 
	the existence of a polynomial-time algorithm, and in the next section we provide numerical evidence of its scalability. The proof of this theorem gives a bound of $\mathcal{O}(|R|^2 \cdot |S|^5)$ on the number of steps. Tighter bounds could be obtained by extending classical partition refinement approaches available for labeled transitions 
	systems~\cite{partitionref,DBLP:journals/jcss/BaierEM00} to CRNs, which is however  the subject of future work.
	
	\section{Evaluation}\label{sec:caseStudies}
	
	We now evaluate \FB{} and \BB{}\@.  We first study their effectiveness in  reducing the ODEs of a number of biochemical models from the literature given in the \texttt{.net} format of BioNetGen~\cite{Blinov22112004}, version 2.2.5-stable.
	Using selected models we discuss how \FB{} and \BB{} relate with each other, and provide a biological interpretation of the aggregations. Finally, we compare them against $\kappa$'s fragmentation. All experiments are replicable using  a prototype available at \url{http://sysma.imtlucca.it/crnreducer/}.
	\newcolumntype{H}{>{\setbox0=\hbox\bgroup}c<{\egroup}@{}}
	\begin{table}[t]
		\centering
		\scalebox{0.79}{
			\begin{tabular}{H lrr rrrr H rrrr}
				\toprule
				&
				\multicolumn{3}{c}{\emph{Original model}} &
				\multicolumn{4}{c}{\emph{Forward reduction}} &
				\emph{WFB} &
				\multicolumn{4}{c}{\emph{Backward reduction}}\\
				\cmidrule(l){1-4} \cmidrule(l){5-8} \cmidrule(l){10-13}
				Model name &
				\multicolumn{1}{l}{\!\emph{Id}} &
				\multicolumn{1}{c}{ \ $| R |$ \ } &
				\multicolumn{1}{c}{ \ $| S |$ \ } &
				\multicolumn{1}{c}{ \quad\ \emph{Red.(s)} } &
				\multicolumn{1}{c}{ \ $| R |$ \ } &
				\multicolumn{1}{c}{ \ $| S |$ \ } &
				\multicolumn{1}{c}{\emph{Speed-up} } &
				$| S |$ &
				\multicolumn{1}{c}{ \quad  \emph{Red.(s)} } &
				\multicolumn{1}{c}{ \ $| R |$ \ } &
				\multicolumn{1}{c}{ \ $| S |$ \ } &
				\multicolumn{1}{c}{\emph{Speed-up} }
				\\
				\midrule
				e9 &
				M1 &  3538944 & 262146 & 4.61E+4  & 990 & 222  &  \multicolumn{1}{c}{ --- } & 222 &  7.65E+4 & 2708 & 222 & \multicolumn{1}{c}{ --- }\\
				e8 &
				M2 &  786432 & 65538  & 1.92E+3  & 720 & 167 &  \multicolumn{1}{c}{ --- }     & 167 &  3.68E+3  & 1950 & 167 &  \multicolumn{1}{c}{ --- } \\
				e7 &
				M3  &  172032 & 16386  & 8.15E+1  & 504 & 122 & 1.16E+3                             & 122 & 1.77E+2  & 1348 & 122 & 5.34E+2 \\
				e2 &
				M4  &  48 & 18  & 1.00E--3 &  24 & 12 &  1.00E+0                                             & 12 &  2.00E--3 &     45 & 12 & 1.00E+0\\
				machine &
				M5 & 194054 & 14531  & 3.72E+1  & 142165 & 10855 & 1.03E+0                     & 10855 &  1.32E+3 & 93033 & 6634 & 1.03E+0\\
				fceri\_gamma2\_asym &
				M6  & 187468 & 10734   & 3.07E+1  & 57508 & 3744 & 1.92E+1                     & 351 &   2.71E+2 & 144473 & 5575 & 3.53E+0\\
				fceri\_fyn\_lig &
				M7 & 32776 & 2506  & 1.26E+0 &   16481 & 1281 & 6.23E+0                          & 154 &  1.66E+1 & 32776 & 2506 & \multicolumn{1}{c}{x}\\
				pcbi.1000364.s006 &
				M8 & 41233 & 2562 & 1.12E+0  & 33075 & 1897 & 1.12E+0                            & 1897 &    1.89E+1 & 41233 & 2562 & \multicolumn{1}{c}{x}\\
				pcbi.1000364.s005 &
				M9  & 5033 & 471 & 1.91E--1 & 4068 & 345 & 1.04E+0                                   & 345 &   4.35E--1 & 5033 & 471     & \multicolumn{1}{c}{x}\\
				pcbi.1003217.s007 &
				M10 & 5797 & 796 & 1.61E--1   & 4210 & 503 & 1.47E+0                                  & 503 &  7.37E--1 & 5797 & 796     & \multicolumn{1}{c}{x}\\
				1471-2105-11-404-s1 &
				M11 & 5832 & 730 & 3.89E--1  & 1296 & 217 & 1.32E+1                                   & 217 &   6.00E--1 & 2434 & 217 & 7.55E+0\\
				scaff22 &
				M12 & 487 & 85 & 2.00E--3  & 264 & 56 &  1.88E+0                                          & 56 &   6.00E--3 & 426 & 56 & 1.31E+0\\
				ncc &
				M13 & 24 & 18 & 1.20E--2  & 24 & 18 &  \multicolumn{1}{c}{x}                           & 1 &  7.00E--3 & 6 & 3 & 1.00E+0 \\
				\bottomrule
			\end{tabular}
		}
		\caption{Forward and backward reductions and corresponding speed-ups in ODE analysis. Speed-up entries ``---'' indicate that the original model could not be solved; entries ``x'' indicate that the coarsest bisimulation did not reduce the original model.}
		\label{table:reductionResults}
	\end{table}
	
	\subparagraph*{Numerical results.}
	Table~\ref{table:reductionResults} lists our case studies:
	%
	four synthetic benchmarks to obtain combinatorially larger CRNs by varying the number of phosphorylation sites 
	(M1--M4)~\cite{citeulike:8493139};
	a model of pheromone signaling (M5,~\cite{10.1371/journal.pcbi.1003278}); two signaling pathways through the Fc$\varepsilon$ complex (M6--M7,~\cite{Faeder01042003,citeulike:8493139});  two models of enzyme activation (M8--M9,~\cite{journals/ploscb/BaruaFH09}); a model of a tumor suppressor protein (M10,~\cite{10.1371/journal.pcbi.1003217}); a model of tyrosine phosphorylation and adaptor protein binding (M11,~\cite{DBLP:journals/bioinformatics/ColvinMFHHP09,DBLP:journals/bmcbi/ColvinMGHHP10});    a MAPK model (M12,~\cite{Kocieniewski2012116}); and an \emph{influence network} (M13,~\cite{CardelliCRN}).
	
	Headers $\abs{R}$ and $\abs{S}$ give the number of reactions and species of the CRN (and of its reductions), respectively. The reduction times (\emph{Red.}) account also for the computation of the quotient  CRNs. The speed-up is the ratio between the time to solve the ODEs of the original CRN and that of the reduced one including the time to reduce the CRN. Measurements were taken on a 2.6 GHz Intel Core i5 with 4\,GB of RAM. The time interval of the ODE solution was taken from the original papers; for M1--M4, where this data was not available, time point 50.0 was used as an estimate of steady state.
	The initial conditions for the ODEs were also taken from the original papers. The initial partition for \FB{} was chosen to be the trivial one containing the singleton block $\{ S \}$ (i.e., no species was singled out). Instead, the initial partition for \BB{} was chosen consistently with the ODE initial conditions; that is, two species may be equivalent only if they have the same initial conditions in the original CRN. This ensured that the backward reduced CRN was a lossless aggregation of the original CRN.

	We make three main observations: (i) \FB{} and \BB{} can reduce a significant number of models. In the two largest models of our case studies, M1 and M2, the bisimulations were able to provide a compact aggregated ODE system which could be straightforwardly analyzed, while the solutions of the original models did not terminate due to out-of-memory errors, consistently with~\cite{citeulike:8493139}. 
	(ii) \FB{} and \BB{} are not comparable in general. For instance, both reduce M5 to $10855$ and $6634$ species, respectively, while M6 is reduced to 3744 species by \FB{}, and to 5574 by \BB{}. Also, \FB{} was able to reduce M7--M10, while \BB{} did not aggregate. The influence network M13 shows the opposite; in fact, none of the influence networks presented in~\cite{CardelliCRN} can be reduced up to \FB{} (here we showed M13, which is the largest one from~\cite{CardelliCRN}). (iii) Models M1--M4 and M12 show that the intersection between \FB{} and \BB{} is nonempty. 
	\subparagraph*{Biological interpretation.}
	Models M1 and M2 enjoy significant reductions and ODE analysis speed-ups. Here we use them to explain that \FB{} and \BB{} are effective at aggregating species representing symmetric sites in a complex. For this, let us consider M4, chosen for space reasons. 
	A typical equivalence class is for instance 
	$\{E(s!1).S(p1\!\!\sim\!\!P,p2\!\!\sim\!\!U!1), E(s!1).S(p1\!\!\sim\!\!U!1,p2\!\!\sim\!\!P)\}$. 
	According to the syntax of the BioNetGen language, the CRN species are formed from basic \emph{molecules} $S$ and $E$. Molecule $S$ has two binding sites ($p1$, and $p2$) which can be either in \emph{phosphorylated} state ($P$) or not ($U$); $E$ has one stateless binding site ($s$) which can bind to those of $S$ to form a  complex. The two sites of $S$ have equivalent capabilities in terms of binding with other species or changing state. For instance, the above equivalence class contains two species composed by $S$ and $E$, with $E$ bound to the unphosphorylated site of $S$  (here the exclamation mark links the binding sites used to form the species). Models M1 and M2 exhibit a fast growth of the number of species due to a larger number of symmetric sites, requiring distinct species to track exactly which site exhibits a particular phosphorylation state.
	This form of symmetry has also been studied in~\cite{Camporesi201129} where the authors propose an 
	approach to detect it directly at the $\kappa$ level.
	However, an experimental comparison could not be performed because~\cite{Camporesi201129} is not yet implemented.
	Although both  bisimulations give the same equivalence classes in these cases, the reduced CRNs have different reactions, since \FB{} provides the dynamics of the sums of equivalent species, while \BB{} considers the distinct dynamics of representative species.
	Instead, aggregation of identical binding sites is supported by BioNetGen. This can be seen in models M6 and M7, since they both have $Lig(l,l)$, a ligand with two copies of site $l$. Intuitively, the rule
	\begin{equation}\label{eq:bnglIdenticalSite}
	Rec(a) + Lig(l,l) \rightarrow Rec(a!1).Lig(l!1,l)
	\end{equation}
	gives rise to only one chemical complex in the underlying CRN, $Rec(a!1).Lig(l!1,l)$. This represents the (forward and backward) canonical representative  of a ligand bound to a single receptor $Rec(a)$. To see this, let us rename the two sites and \emph{expand} the rule appropriately: 
	
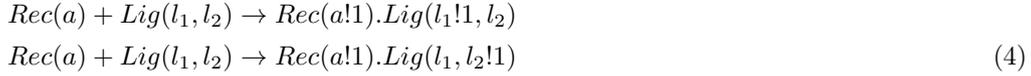
\begin{equation}\label{eq:expansionIdenticalSite}
	\begin{split}
	Rec(a) + Lig(l_1,l_2) & \rightarrow Rec(a!1).Lig(l_1!1,l_2) \\
	Rec(a) + Lig(l_1,l_2) & \rightarrow Rec(a!1).Lig(l_1,l_2!1)
	\end{split}
	\end{equation}
	Then, this underlying CRN will distinguish the two sites.  However, applying either of our CRN bisimulations leads to the CRN for Equation~\eqref{eq:bnglIdenticalSite}.
	
	We remark that the original CRN sizes of M6 and M7 already account for the aggregations obtained with BioNetGen. Nevertheless, our CRN bisimulations allow for further (significant) reductions.
	For instance, part of the reductions for M6 are due to the presence of $Rec(a,b,g_1,g_2)$,  a molecule  with symmetric sites $g_1$ and $g_2$, similarly to those of M4.

	Symmetric sites are not the only property captured by our bisimulations. For instance in both M8 and M9 one of the \FB{} equivalence classes 
	is given by:
	\begin{align*}
	&\big\{ \ \! J(k!1).R(x!1,i\!\!\sim\!\!\text{on},l), \ J(k!1).L(r_1!2,r_2).R(x!1,i\!\!\sim\!\!\text{on},l!2), \\
	& \ \ J(k!1).L(r_1,r_2!2).R(x!1,i\!\!\sim\!\!\text{on},l!2), \\
	&\ \ J(k!1).L(r_1!2,r_2!3).R(x!1,i\!\!\sim\!\!\text{on},l!3).R(x,i\!\!\sim\!\!\text{on},l!2)\big \} .
	\end{align*}
	%
	A biological interpretation is that a species containing the molecule $J$ behaves in the same way as long as it is bound to a molecule $R$ having binding site $i$ in state ``on''. This is independent of whether $R$  is further complexed with other molecules via its binding site $l$; For instance, the first species models that $R$ is only bound to $J$, while in the second and third species it is also bound to $L$.
	Finally, in M5, one of the \BB{} equivalence classes is 
	\begin{align*}
	\big\{ & \text{Dig2(p!1).Ste12(dig1,dig2,dna!1,mapk)},\ \text{Fus3(p!1).Ste12(dig1,dig2,dna!1,mapk)}, \\
	& \text{Msg5(p!1).Ste12(dig1,dig2,dna!1,mapk)},\ \text{Sst2(p!1).Ste12(dig1,dig2,dna!1,mapk)}, \\
	& \text{Ste12(p!1).Ste12(dig1,dig2,dna!1,mapk)},\ \text{Ste2(p!1).Ste12(dig1,dig2,dna!1,mapk)} \big\} .
	\end{align*}
	It captures that genes Dig2, Fus3, Msg5, Sst2, Ste12, and Ste2, bind to the  protein Ste12 with equal rates. This yields equivalent  dynamics for these Ste12-gene complexes, and all those formed by them which are equal up to the gene bound to Ste12.
	
	
	
	\subparagraph*{Experimental comparison with $\kappa$-based reduction techniques.}
	We now experimentally compare our CRN bisimulations and fragmentation in the case of rule-based biochemical models for which the underlying CRN can be fully enumerated. All models in \autoref{table:reductionResults} belong to this class; however, none of them was originally available in $\kappa$, the only language that supports fragmentation. Thus, we performed a manual translation of a selection of the case studies from the BioNetGen language into $\kappa$.
	\footnote{
		\messageAboutK}
	We found:
	
	\begin{itemize}
		\item
		\emph{Models that can be reduced by CRN bisimulations but not by fragmentation.} 
		The $\kappa$ encoding of M12 (a case where only cosmetic syntactical changes are required) returned 85 fragments, equal to the size of the CRN, while both \FB{} and \BB{} reduced to 56 species. 
		The encodings of M6 and M7 necessitated expansions analogous to Equation~\eqref{eq:expansionIdenticalSite} because $\kappa$ does not currently
		support distinct sites with the same name.
		This led to bigger initial CRNs, for which fragmentation returned 58040 fragments for M6 and 10930 for M7.
		%
		%
		\item
		\emph{Models that can be reduced by fragmentation but not by our bisimulations.}
		The $\kappa$ model of early events of the EGF pathway in~\cite{4120004} is reduced from 356 species to 38 fragments~\cite{DBLP:conf/lics/DanosFFHK10}, while no aggregation is obtained with either \FB{} or \BB{}.
		%
		\item
		\emph{Models that can be reduced by both our bisimulations and fragmentation.} The $\kappa$ encodings of models M1--M4 present different reductions than using either bisimulation, specifically 38, 34, 30 and 10 fragments (versus 222, 167, 122, and 12 \FB{} and \BB{} equivalence classes, respectively). It can be shown that, in the latter examples, the reductions are complementary, in the sense that no two bisimilar species are included in the same fragment. While our bisimulations captured symmetric sites,  
		fragments explain that the sites of $S$ are \emph{independent}, i.e., the state of a site does not affect the dynamics of the other. For instance, one of the fragments for model M4 is
		\[
		\{S(p1\!\!\sim\!\!P,p2\!\!\sim\!\!P) , S(p1\!\!\sim\!\!P,p2\!\!\sim\!\!U) ,  E(s!1).S(p1\!\!\sim\!\!P,p2\!\!\sim\!\!U!1) ,   F(s!1).S(p1\!\!\sim\!\!P,p2\!\!\sim\!\!P!1)\}
		\]
		which essentially collects all species where the $p1$ site of molecule $S$ is phosphorylated.
	\end{itemize}
	
	
\section{Conclusion}
Forward and backward bisimulations are equivalence relations over the species of a chemical reaction network inducing a partition of the underlying mass-action system of ordinary differential equations. An experimental evaluation has demonstrated their usefulness  by showing their complementarity as well as significant model reductions in a number of biochemical models available in the literature. This has been supported by a prototype, which currently allows a ready-to-use tool-chain with BioNetGen, a state-of-the-art tool.

Ongoing work is studying stochastic counterparts of both forward and backward bisimulations, to obtain model reductions when the semantics of chemical reaction networks based on continuous-time Markov chains is considered.  Also, we plan to investigate the applicability of our bisimulations in other model repositories, e.g., those using the well-known SBML interchange format (\url{http://sbml.org}).

	\section*{Acknowledgement}
	The authors thank J. Krivine and J. Feret for helpful discussions, and the anonymous referees
	for suggestions that improved the paper.
	
	This work was partially supported by the EU project QUANTICOL, 600708. L. Cardelli is partially funded by a Royal Society Research Professorship. Part of this research has been carried out while M. Tribastone, M. Tschaikowski, and A. Vandin were at University of Southampton, UK.
	\bibliography{concur2015}

\begin{thebibliography}{10}

\bibitem{DBLP:journals/jcss/BaierEM00}
C.~Baier, B.~Engelen, and M.~E. Majster-Cederbaum.
\newblock Deciding bisimilarity and similarity for probabilistic processes.
\newblock {\em J. Comput. Syst. Sci.}, 60(1):187--231, 2000.

\bibitem{journals/ploscb/BaruaFH09}
D.~Barua, J.~R. Faeder, and J.~M. Haugh.
\newblock A bipolar clamp mechanism for activation of jak-family protein
  tyrosine kinases.
\newblock {\em PLoS Computational Biology}, 5(4), 2009.

\bibitem{10.1371/journal.pcbi.1003217}
D.~Barua and W.~S. Hlavacek.
\newblock Modeling the effect of apc truncation on destruction complex function
  in colorectal cancer cells.
\newblock {\em PLoS Comput Biol}, 9(9):e1003217, 09 2013.

\bibitem{Bernardo:2007ij}
M.~Bernardo.
\newblock A survey of {Markovian} behavioral equivalences.
\newblock In {\em Formal Methods for Perf. Eval.}, volume 4486 of {\em LNCS},
  pages 180--219. Springer Berlin Heidelberg, 2007.

\bibitem{Blinov22112004}
M.~L. Blinov, J.~R. Faeder, B.~Goldstein, and W.~S. Hlavacek.
\newblock {BioNetGen}: software for rule-based modeling of signal transduction
  based on the interactions of molecular domains.
\newblock {\em Bioinformatics}, 20(17):3289--3291, 2004.

\bibitem{4120004}
M.~L. Blinov, J.~R. Faeder, B.~Goldstein, and W.~S. Hlavacek.
\newblock {A network model of early events in epidermal growth factor receptor
  signaling that accounts for combinatorial complexity}.
\newblock {\em Biosystems}, 83:136--151, 2006.

\bibitem{BuchholzOrdinaryExact}
P.~Buchholz.
\newblock {Exact and Ordinary Lumpability in Finite Markov Chains}.
\newblock {\em Journal of Applied Probability}, 31(1):59--75, 1994.

\bibitem{bucholz94papm}
P.~Buchholz.
\newblock {Markovian Process Algebra: Composition and Equivalence}.
\newblock In {\em Proc. 2nd Workshop on Process Algebra and Performance
  Modelling}, Erlangen, Germany, 1994.

\bibitem{Camporesi201129}
F.~Camporesi and J.~Feret.
\newblock Formal reduction for rule-based models.
\newblock {\em Electronic Notes in Theoretical Computer Science}, 276:29--59,
  2011.
\newblock {MFPS XXVII}.

\bibitem{DBLP:journals/entcs/CamporesiFKP10}
Ferdinanda Camporesi, J{\'{e}}r{\^{o}}me Feret, Heinz Koeppl, and Tatjana
  Petrov.
\newblock Combining model reductions.
\newblock {\em Electr. Notes Theor. Comput. Sci.}, 265:73--96, 2010.

\bibitem{CardelliCRN}
L.~Cardelli.
\newblock Morphisms of reaction networks that couple structure to function.
\newblock {\em BMC Systems Biology}, 8(1):84, 2014.

\bibitem{Cardelli:2012aa}
L.~Cardelli and A.~Csik{\'a}sz-Nagy.
\newblock The cell cycle switch computes approximate majority.
\newblock {\em Sci. Rep.}, 2, 2012.

\bibitem{DBLP:journals/bioinformatics/ColvinMFHHP09}
J.~Colvin, M.~I. Monine, J.~R. Faeder, W.~S. Hlavacek, D.~D.~Von Hoff, and
  R.~G. Posner.
\newblock Simulation of large-scale rule-based models.
\newblock {\em Bioinformatics}, 25(7):910--917, 2009.

\bibitem{DBLP:journals/bmcbi/ColvinMGHHP10}
J.~Colvin, M.~I. Monine, R.~N. Gutenkunst, W.~S. Hlavacek, D.~D.~Von Hoff, and
  R.~G. Posner.
\newblock Rulemonkey: software for stochastic simulation of rule-based models.
\newblock {\em {BMC} Bioinformatics}, 11:404, 2010.

\bibitem{16430778}
H.~Conzelmann, J.~Saez-Rodriguez, T.~Sauter, B.~Kholodenko, and E.~Gilles.
\newblock A domain-oriented approach to the reduction of combinatorial
  complexity in signal transduction networks.
\newblock {\em BMC Bioinformatics}, 7(1):34, 2006.

\bibitem{DBLP:conf/lics/DanosFFHK10}
V.~Danos, J.~Feret, W.~Fontana, R.~Harmer, and J.~Krivine.
\newblock Abstracting the differential semantics of rule-based models: Exact
  and automated model reduction.
\newblock In {\em LICS}, pages 362--381, 2010.

\bibitem{Danos200469}
V.~Danos and C.~Laneve.
\newblock Formal molecular biology.
\newblock {\em TCS}, 325(1):69--110, 2004.

\bibitem{De-Nicola:1990aa}
R.~De~Nicola, U.~Montanari, and F.~Vaandrager.
\newblock Back and forth bisimulations.
\newblock In {\em CONCUR}, volume 458 of {\em LNCS}, pages 152--165. Springer,
  1990.

\bibitem{Faeder01042003}
J.~R. Faeder, W.~S. Hlavacek, I.~Reischl, M.~L. Blinov, H.~Metzger, A.~Redondo,
  C.~Wofsy, and B.~Goldstein.
\newblock Investigation of early events in {Fc$\varepsilon$RI}-mediated
  signaling using a detailed mathematical model.
\newblock {\em The Journal of Immunology}, 170(7):3769--3781, 2003.

\bibitem{Feret2012137}
J.~Feret, T.~Henzinger, H.~Koeppl, and T.~Petrov.
\newblock Lumpability abstractions of rule-based systems.
\newblock {\em TCS}, 431:137--164, 2012.

\bibitem{Feret_IJSI2013}
Jerome Feret, Heinz Koeppl, and Tatjana Petrov.
\newblock Stochastic fragments: A framework for the exact reduction of the
  stochastic semantics of rule-based models.
\newblock {\em International Journal of Software and Informatics}, 7(4):527 --
  604, 2013.

\bibitem{export:68022}
J.~Fisher and T.A. Henzinger.
\newblock Executable cell biology.
\newblock {\em Nature Biotechnology}, 25(11):1239--1249, 2007.
\newblock See also correspondence in Nature Biotechnology 26(7):737-8;738-9,
  2008.

\bibitem{citeulike:1079741}
D.~Gillespie.
\newblock {The chemical Langevin equation}.
\newblock {\em The Journal of Chemical Physics}, 113(1):297--306, 2000.

\bibitem{Heath2008239}
J.~Heath, M.~Kwiatkowska, G.~Norman, D.~Parker, and O.~Tymchyshyn.
\newblock Probabilistic model checking of complex biological pathways.
\newblock {\em TCS}, 391(3):239--257, 2008.

\bibitem{hermanns:mtipp}
H.~Hermanns and M.~Rettelbach.
\newblock Syntax, semantics, equivalences, and axioms for {MTIPP}.
\newblock In {\em Proceedings of Process Algebra and Probabilistic Methods},
  pages 71--87, Erlangen, 1994.

\bibitem{pepa}
J.~Hillston.
\newblock {\em A Compositional Approach to Performance Modelling}.
\newblock CUP, 1996.

\bibitem{Kocieniewski2012116}
P.~Kocieniewski, J.~R. Faeder, and T.~Lipniacki.
\newblock The interplay of double phosphorylation and scaffolding in {MAPK}
  pathways.
\newblock {\em Journal of Theoretical Biology}, 295:116--124, 2012.

\bibitem{Larsen19911}
K.~G. Larsen and A.~Skou.
\newblock Bisimulation through probabilistic testing.
\newblock {\em Information and Computation}, 94(1):1--28, 1991.

\bibitem{okino1998}
M.~S. Okino and M.~L. Mavrovouniotis.
\newblock Simplification of mathematical models of chemical reaction systems.
\newblock {\em Chemical Reviews}, 2(98):391--408, 1998.

\bibitem{partitionref}
R.~Paige and R.~Tarjan.
\newblock Three partition refinement algorithms.
\newblock {\em SIAM Journal on Computing}, 16(6):973--989, 1987.

\bibitem{Shin14}
S.~W. Shin, C.~Thachuk, and E.~Winfree.
\newblock Verifying chemical reaction network implementations: A pathway
  decomposition approach.
\newblock In {\em VEMPD}, Vienna Summer of Logic, 2014.

\bibitem{citeulike:8493139}
M.~W. Sneddon, J.~R. Faeder, and T.~Emonet.
\newblock Efficient modeling, simulation and coarse-graining of biological
  complexity with {NFsim}.
\newblock {\em Nature Methods}, 8(2):177--183, 2011.

\bibitem{1703385}
J.~Sproston and S.~Donatelli.
\newblock {Backward Bisimulation in Markov Chain Model Checking}.
\newblock {\em IEEE Trans. Software Eng.}, 32(8):531--546, 2006.

\bibitem{10.1371/journal.pcbi.1003278}
R.~Suderman and E.~J. Deeds.
\newblock Machines vs. ensembles: Effective {MAPK} signaling through
  heterogeneous sets of protein complexes.
\newblock {\em PLoS Comput Biol}, 9(10):e1003278, 10 2013.

\bibitem{LiRabitz1997}
J.~Toth, G.~Li, H.~Rabitz, and A.~S. Tomlin.
\newblock The effect of lumping and expanding on kinetic differential
  equations.
\newblock {\em SIAM Journal on Applied Mathematics}, 57(6):1531--1556, 1997.

\bibitem{concur2012}
M.~Tschaikowski and M.~Tribastone.
\newblock {Exact fluid lumpability for Markovian process algebra}.
\newblock In {\em CONCUR}, LNCS, pages 380--394, 2012.

\bibitem{jlamp14}
M.~Tschaikowski and M.~Tribastone.
\newblock {A unified framework for differential aggregations in Markovian
  process algebra}.
\newblock {\em JLAMP}, 84(2):238--258, 2015.

\end{thebibliography}
		
\appendix
\section*{APPENDIX}\label{sec:appendix}
\subsection*{A.1\ Forward CRN Bisimulation}

The following auxiliary lemma will be needed in the proofs of Propositions~\ref{prop:cofsb:ex} and~\ref{prop:cefsb:ex}.

\begin{lemma}\label{lem_ex_cor_aux}
	Let $I$ be an index set and let $\R_i$ denote equivalence relations on $S$ such that $S / \R_i$ is a refinement of some partition $\mathcal{G}$ of $S$. Then, $\R = ( \bigcup_{i \in I} \R_i)^\ast$ is such that $S / \R$ is a refinement of $\mathcal{G}$.
\end{lemma}

\begin{proof}
	We first note that $\R$ is an equivalence relation over $S$, as it is the transitive closure of the union of equivalence relations over $S$. For $i \in I$, set $\pa_i = S/ \R_i$ and $\pa = S / \R$. Then, for any $(y_0,y_1) \in \R$, there exist $x_0,\ldots,x_k \in S$ such that $x_0 R_{i_0} x_1 R_{i_1} \ldots R_{i_{k-1}} x_k$ with $y_0 = x_0$, $y_1 = x_k$ and $i_j \in I$ for all $0 \leq j \leq k - 1$. Moreover, $x_0 \in G$ for some (unique) $G \in \mathcal{G}$. We show that $x_0, \ldots, x_k \in G$ by induction. Since the base case $j = 0$ is trivial, let us consider the induction step $j - 1 \to j$. Then, $x_{j - 1} \R_{i_{j - 1}} x_j$ implies the existence of some $H \in \mathcal{H}_{i_{j - 1}}$ such that $x_{j - 1}, x_j \in H$. Let $G_{j - 1} \in \mathcal{G}$ be such that $H \subseteq G_{j - 1}$. Since $x_{j - 1} \in G$ by induction hypothesis and $x_{j - 1} \in G_{j - 1}$, it holds that $G \cap G_{j - 1} \neq \emptyset$. Since $\mathcal{G}$ is a partition, this implies that $G = G_{j - 1}$, yielding in turn $x_j \in G$.
\end{proof}


\begin{proof}[Proof of Proposition~\ref{prop:cofsb:ex}]
	We first note that $\R$ is an equivalence relation over $S$, as it is the transitive closure of the union of equivalence relations over $S$. For $i \in I$, set $\pa_i = S/ \R_i$ and $\pa = S / \R$. Note that, for all $i \in I$, any $\Y \in \pa$ is a union of blocks of $\pa_i$. For any $(y_0,y_1) \in \R$, there exist $x_0,\ldots,x_k \in S$ such that $x_0 R_{i_0} x_1 R_{i_1} \ldots R_{i_{k-1}} x_k$ with $y_0 = x_0$, $y_1 = x_k$ and $i_j \in I$ for all $0 \leq j \leq k - 1$. Noting that $x_j R_{i_j} x_{j + 1}$ for all $0 \leq j \leq k - 1$, we infer that $\crr[x_j,\rho] = \crr[x_{j + 1},\rho]$ and $\gr[x_j,\rho,\Y] = \gr[x_{j + 1},\rho,\Y]$ for all $\rho \in \ms(S)$ and $\Y \in \pa$. Thus, we infer that $\crr[x_0,\rho] = \crr[x_1,\rho] = \ldots = \crr[x_{k - 1},\rho] = \crr[x_k,\rho]$ and $\gr[x_0,\rho,\Y] = \gr[x_1,\rho,\Y] = \ldots = \gr[x_{k - 1},\rho,\Y] = \gr[x_k,\rho,\Y]$. The remainder of the claim follows then from Lemma~\ref{lem_ex_cor_aux}.
\end{proof}

We now provide a sort of commutative property for $\crr$ and $\gr$ for the case of singleton $\rho$, used to prove Theorem~\ref{th:dsbAndODELumping}.   

%

\begin{proposition}\label{prop:commCRR}
	Let $(S,R)$ be a CRN and $X,Y\in S$.
	Then it holds
	\begin{align}
	\crr(X,Y) & = \crr(Y,X)\label{eq:c}\\
	\gr(X,Y,Z) & = \gr(Y,X,Z)\label{eq:p}
	\end{align}
\end{proposition}

\begin{proof}
	The statement follows from Definition~\ref{def:conditionalReactionRate}, since \eqref{eq:c} and~\eqref{eq:p} can be rewritten, resp., as:
	\begin{align*}
	\qquad\qquad\
	\multiset{X,Y}(X)\cdot\!\!\!\sum_{X+Y\act{\rate} \pi \in R}  \rate &= \multiset{X,Y}(Y)\cdot\!\!\!\sum_{X+Y\act{\rate} \pi \in R}  \rate\\
	\multiset{X,Y}(X)\cdot\!\!\!\sum_{X+Y\act{\rate} \pi \in R}  \rate \cdot \pi(Z) &= \multiset{X,Y}(Y)\cdot\!\!\! \sum_{X+Y\act{\rate} \pi \in R}  \rate \cdot \pi(Z)
	\end{align*}
	%
\end{proof}

\paragraph*{Proof of Theorem~\ref{th:dsbAndODELumping}}

For $(S,R)$ a CRN, $\Y\subseteq S$, and $V$ a concentration function for $(S,R)$,  we use $V_\Y$ for $\sum_{X \in \Y}\!V_X$, and \!$X^\Y$\! or \!$Y^\Y$\! to denote any element of $\Y$.
We first separate the influence exerted by each reaction to the concentration of a species in a negative one (the \emph{depletion} rate), and in a positive one (the \emph{accretion rate}).

\begin{remark}\label{rm:accrDepl}
	Let $(S,R)$ be a CRN, $V$ a concentration function
	for $(S,R)$, and $X\in S$. We denote the \emph{depletion} and \emph{accretion} rates of $X$ due to a reaction of $R$ as, respectively,
	\begin{align*}
	Depl(\rho \act{\rate} \pi,X,V) & = \rho(X) \cdot \rate \cdot  \prod_{Y \in S} V_Y^{\rho(Y)} \\
	Accr(\rho \act{\rate} \pi,X,V) & = \pi(X) \cdot \rate \cdot  \prod_{Y \in S} V_Y^{\rho(Y)}
	\end{align*}
	We can then reformulate the definition of vector field of $(S,R)$ given in Section~\ref{sec:crn} as:
	\begin{align*}
	\ \ \ \ \quad\quad\qquad F_X(V) = \sum_{\rho \act{\rate} \pi \in R} \Big(Accr(\rho \act{\rate} \pi,X,V) - Depl(\rho \act{\rate} \pi,X,V)\Big) \qquad\quad\quad\ \ \ \  \qed
	\end{align*}
\end{remark}

We now state that the aggregated accretion and depletion rates of an equivalence class of an \FB{} can be written in terms of the aggregate concentrations.
\begin{proposition}[Aggregated depletion and accretion rate]\label{prop:aggregDeplAndAccr}
	Let $(S,R)$ be a CRN, $V$ a concentration function for $(S,R)$, and
	$\pa$ an \FB{}. Then, for any $\Y \in \pa$
	\[\sum_{X \in \Y} \sum_{\rho\act{\rate}\pi \in R} Depl(\rho\act{\rate}\pi,V,X) \quad \text{ and } \quad
	\sum_{X \in \Y} \sum_{\rho\act{\rate}\pi \in R} Accr(\rho\act{\rate}\pi,V,X)\]
	can be written in terms of the aggregated concentrations
	$V_{\Y_1}\ldots V_{\Y_n}$. 
\end{proposition}

\begin{proof}

	We start addressing the depletion rate case. 
	We have
	\begin{align}
	\sum_{X \in \Y} \sum_{\rho \act{\rate} \pi \in R} Depl(\rho \act{\rate} \pi,V,X)&  =
	\sum_{X \in \Y} \sum_{\rho\in\MS(S)}\sum_{\rho \act{\rate} \pi \in R} \rho(X) \cdot \rate \cdot \prod_{Y \in S} V_Y^{\rho(Y)} \notag \\
	& = \sum_{X \in \Y} \sum_{\rho\in\MS(S)}\prod_{Y \in S} V_Y^{\rho(Y)}\sum_{\rho \act{\rate} \pi \in R} \rho(X) \cdot \rate \label{eq:aggregDepl0}
	\end{align}
	Given that we have unary or binary reactions only, we can rewrite (\ref{eq:aggregDepl0}) as:
	\begin{align}
	&\sum_{X \in \Y} \sum_{Y\in S}V_Y\sum_{Y \act{\rate} \pi \in R} \multiset{Y}(X) \cdot \rate \mathop{+}\label{eq:aggregDepl1}\\
	&\quad\sum_{X \in \Y} \sum_{Y+Y'\in\MS(S)}V_Y\cdot V_Y'\sum_{Y+Y' \act{\rate} \pi \in R} \multiset{Y,Y'}(X) \cdot \rate \label{eq:aggregDepl2}
	\end{align}
	Now, (\ref{eq:aggregDepl1}) can be easily rewritten in terms of the aggregated variables only:
	\begin{align*}
	\sum_{X \in \Y} \sum_{Y\in S}V_Y\sum_{Y \act{\rate} \pi \in R} \multiset{Y}(X) \cdot \rate & =
	\sum_{X \in \Y} V_X\sum_{X \act{\rate} \pi \in R} \rate  =  \sum_{X \in \Y} V_X\cdot\crr[X,\emptyset],
	\end{align*}
	from which by the condition on $\crr$ of Definition~\ref{def:dsb} we obtain
	$\crr[X^\Y,\emptyset]\cdot V_\Y$.
	
	Instead, as regards (\ref{eq:aggregDepl2}) we have:
	\begin{align*}
	&\sum_{X \in \Y} \sum_{Y+Y'\in\MS(S)} V_Y\cdot V_Y'\sum_{Y+Y' \act{\rate} \pi \in R} \multiset{Y,Y'}(X) \cdot \rate=\\
	&\sum_{X \in \Y} \sum_{X+Y\in\MS(S)}V_X\cdot V_Y\sum_{X+Y \act{\rate} \pi \in R} \multiset{X,Y}(X) \cdot \rate=\\
	&\sum_{X \in \Y}  V_X\sum_{X+Y\in\MS(S)} V_Y\cdot\crr[X,Y]=\comInAlign{by the condition on $\crr$ of \text{Definition~\ref{def:dsb}}}\\
	&\sum_{X \in \Y}  V_X\sum_{X+Y\in\MS(S)} V_Y\cdot\crr[X^\Y,Y]=\\
	&\sum_{X \in \Y}  V_X\sum_{\tilde\Y\in\pa}\sum_{Y\in\tilde \Y} V_Y\cdot\crr[X^\Y,Y]=\comInAlign{by \text{Proposition~\ref{prop:commCRR}}}\\
	&\sum_{X \in \Y}  V_X\sum_{\tilde\Y\in\pa}\sum_{Y\in\tilde \Y} V_Y\cdot\crr[Y,X^\Y]=\comInAlign{by the condition on $\crr$ of \text{Definition~\ref{def:dsb}}}\\
	&\sum_{X \in \Y}  V_X\sum_{\tilde\Y\in\pa}\crr[Y^{\tilde\Y},X^\Y]\sum_{Y\in\tilde \Y} V_Y=\\
	&\sum_{X \in \Y}  V_X\sum_{\tilde\Y\in\pa}\crr[Y^{\tilde\Y},X^\Y]\cdot V_{\tilde\Y}=\\
	&\sum_{\tilde\Y\in\pa} V_{\tilde\Y} \cdot\crr[Y^{\tilde\Y},X^\Y] \sum_{X \in \Y} V_X=\\
	&\sum_{\tilde\Y\in\pa} V_{\tilde\Y} \cdot\crr[Y^{\tilde\Y},X^\Y] \cdot V_\Y= V_\Y\sum_{\tilde\Y\in\pa} V_{\tilde\Y} \cdot\crr[X^\Y,Y^{\tilde\Y}]
	\end{align*}
	closing the case.

We now address the accretion rate case. We have
\\
\begin{align}
\sum_{X \in \Y} \sum_{\rho \act{\rate} \pi \in R} Accr(\rho \act{\rate} \pi,V,X) & =
\sum_{X \in \Y} \sum_{\rho \in \ms(S)} \sum_{\rho \act{\rate} \pi \in R} \pi(X) \cdot \rate \cdot \prod_{Y \in S} V_Y^{\rho(Y)} \notag\\
& = \sum_{X \in \Y} \sum_{\rho \in \ms(S)} \prod_{Y \in S} V_Y^{\rho(Y)} \sum_{\rho \act{\rate} \pi \in R} \pi(X) \cdot \rate \label{eq:aggregAccr}
\end{align}
Given that we have unary or binary reactions only, we can rewrite (\ref{eq:aggregAccr}) as:
\begin{align}
\sum_{X \in \Y} \sum_{Y \in S} V_Y \sum_{Y \act{\rate} \pi \in R} \pi(X) \cdot \rate \mathop{+}\label{eq:aggergAccr1} \qquad \qquad \qquad \qquad \qquad \qquad \\
\quad  \mathop{+} \sum_{X \in \Y} \sum_{Y+Y' \in \ms(S)} V_Y\cdot  V_{Y'} \sum_{Y+Y' \act{\rate} \pi \in R} \pi(X) \cdot \rate
\label{eq:aggergAccr2}
\end{align}
Now, (\ref{eq:aggergAccr1}) can be easily rewritten in terms of the aggregated variables only: 
\begin{align*}
&\sum_{Y \in S} V_Y \sum_{X \in \Y} \sum_{Y \act{\rate} \pi \in R} \pi(X) \cdot \rate =\sum_{Y \in S} V_Y \cdot \gr[Y,\emptyset,\Y] =
\sum_{\tilde \Y \in \pa}\sum_{Y \in \tilde \Y} V_Y \cdot\gr[Y,\emptyset,\Y]\;,
\end{align*}
from which, by the condition on $\gr$  of Definition~\ref{def:dsb}, we obtain
$\sum_{\tilde \Y \in \pa}\gr[Y^{\tilde \Y},\emptyset,\Y]\cdot  V_{\tilde \Y}$\ .
Instead, as regards (\ref{eq:aggergAccr2}) we have:
	\begin{align}
	&\sum_{{Y+Y'} \in \ms(S)} V_Y\cdot V_{Y'}\sum_{X \in \Y} \sum_{{Y+Y'} \act{\rate} \pi \in R} \pi(X) \cdot \rate=\notag\\
	&\sum_{{Y+Y} \in \ms(S)} V_Y^2 \sum_{X \in \Y} \sum_{{Y+Y} \act{\rate} \pi \in R} \pi(X) \cdot \rate +\notag\\
	&\quad\sum_{{{Y+Y'} \in \ms(S)}\ s.t.\ {Y \not = Y'}} V_Y\cdot V_{Y'}\sum_{X \in \Y}\sum_{{Y+Y'} \act{\rate} \pi \in R} \pi(X) \cdot \rate =\notag\quad\comInAlign{see below}\\
	&\frac{1}{2}\sum_{Y \in S} V_Y^2 \sum_{X \in \Y} 2\sum_{{Y+Y} \act{\rate} \pi \in R} \pi(X) \cdot \rate +
	\label{eq:twovoertwo}\\
	&\quad\frac{1}{2}\sum_{Y \in S}V_Y\sum_{{Y' \in S}\ s.t.\ {Y \not = Y'}} V_{Y'} \sum_{X \in \Y}\sum_{{Y+Y'} \act{\rate}
		\pi \in R} \pi(X) \cdot \rate = \quad 
	\label{eq:half}\\
	&\frac{1}{2}\sum_{Y \in S}  V_Y^2 \cdot\gr[Y,Y,\Y] +
	\frac{1}{2}\sum_{Y \in S} V_Y\sum_{{Y' \in S}\ s.t.\ {Y \not = Y'}} V_{Y'}
	\cdot\gr[Y,{Y'},\Y]=\notag\\
	&\frac{1}{2}\sum_{Y \in S}  V_Y\sum_{Y' \in S}  V_{Y'} \cdot\gr[Y,{Y'},\Y]=\notag\\
	& \frac{1}{2}\sum_{\tilde \Y \in \pa}\sum_{Y \in \tilde \Y} V_Y {\sum_{Y' \in S}  V_{Y'} \cdot\gr[Y,{Y'},\Y]}=\qquad\text{\emph{(by cond on $\gr$ of \text{Definition~\ref{def:dsb}})}}\notag\\
	%
	& \frac{1}{2}\sum_{\tilde \Y \in \pa}V_{\tilde \Y}\sum_{Y' \in S} V_{Y'} \cdot\gr[Y^{\tilde \Y},{Y'},\Y]=\notag\\
	& \frac{1}{2}\sum_{\tilde \Y \in \pa}V_{\tilde \Y}\sum_{\widehat \Y \in \pa}\sum_{Y' \in \widehat \Y} V_{Y'} \cdot\gr[Y^{\tilde \Y},{Y'},\Y]=\quad\text{\emph{(by \text{Proposition~\ref{prop:commCRR}})}}\notag\\
	& \frac{1}{2}\sum_{\tilde \Y \in \pa}V_{\tilde \Y}\sum_{\widehat \Y \in \pa}\sum_{Y' \in \widehat \Y} V_{Y'} \cdot\gr[Y',Y^{\tilde \Y},\Y]=\quad\text{\emph{(by cond on $\gr$ of \text{Definition~\ref{def:dsb}})}}\notag\\
	& \frac{1}{2}\sum_{\tilde \Y \in \pa}V_{\tilde \Y}\sum_{\widehat \Y \in \pa}\gr[Y^{\widehat{\Y}},{Y^{\tilde \Y}},\Y]\sum_{Y' \in \widehat \Y} V_{Y'}= \frac{1}{2}\sum_{\tilde \Y \in \pa}V_{\tilde \Y}\sum_{\widehat \Y \in \pa}V_{\widehat \Y}\cdot\gr[Y^{\widehat{\Y}},{Y^{\tilde \Y}},\Y]\notag
	%
	\end{align}
Where (\ref{eq:twovoertwo}) is multiplied by $2/2$, while (\ref{eq:half}) is divided by two because each multiset is considered twice (e.g., $Y+Y'$ and $Y'+Y$).
\end{proof}

Finally we can provide the proof of Theorem~\ref{th:dsbAndODELumping}.

\begin{proof}[Proof of Theorem~\ref{th:dsbAndODELumping}]
	The theorem can be proved by showing that, for each $\Y\in \pa$, the sum of the ODEs of the species in $\Y$ can be expressed in terms of the aggregated concentrations of the blocks of $\pa$ only. Such sum is:
	\begin{align*}
	&\sum_{X \in \Y} \sum_{\rho\act{\rate}\pi \in R} Accr(\rho\act{\rate}\pi, V,X) -
	\sum_{X \in \Y} \sum_{\rho\act{\rate}\pi \in R} Depl(\rho\act{\rate}\pi, V,X)
	\end{align*}
	Thus, the claim follows from Proposition~\ref{prop:aggregDeplAndAccr}. 
\end{proof}

\paragraph*{Proof of Theorem~\ref{th:reducedModelHasLumpedODEs}}
Given a CRN $(S,R)$ and an \FB{} $\pa$, we hereby provide the technical results relating the $\pa$-lumped ODEs 
of $(S,R)$ and the ODEs 
of its $(\pa,F)$-reduction. 
In the following we use $\canonical{\Y}$ to denote the canonical representative of the species in a block $\Y\in\pa$.
In addition, for $X\in S$ and $\pi\in\ms(S)$, we may use  $\mu_X$ for $\mu(X)$ and  $\mu_\pi$ for $\mu(\pi)$.
Finally, given a concentration function $V$ for $(S,R)$, we use $\reduce{V}{(\pa,F)}$ for the corresponding concentration function for $\reduce{(S,R)}{(\pa,F)}$, having a component $V_{\canonical{\Y}}^{(\pa,F)} = V_\Y$ per block $\Y\in\pa$. 

We start providing a proposition used in the proof of Theorem~\ref{th:reducedModelHasLumpedODEs}.

\begin{proposition}\label{prop:aggrAccrRedAccr}
	Let $(S,R)$ be a CRN, $\R$ an \FB{}, $\pa = S / \R$, and $\mu$ its choice function.
	Let $V$ be a concentration function for $(S,R)$. 
	Then, for any $\Y\in\pa$ we have
	\begin{align}\sum_{X\in\Y}\sum_{\rho \act{\rate} \pi \in R} Accr(\rho \act{\rate} \pi,X, V) =
	\sum_{\stackrel{\rho \act{\rate} \pi \in R}{\rho=\mu(\rho)}} Accr(\rho \act{\rate} \mu(\pi),\canonical{\Y},\reduce{V}{(\pa,F)})\label{eq:aggrAccrRedAccrInProp}\\
	\sum_{X\in\Y}\sum_{\rho \act{\rate} \pi \in R} Depl(\rho \act{\rate} \pi,X, V) = \sum_{\stackrel{\rho \act{\rate} \pi \in R}{\rho=\mu(\rho)}} Depl(\rho \act{\rate} \mu(\pi),\canonical{\Y},\reduce{V}{(\pa,F)})\label{eq:aggrDeplRedDeplInProp}
	\end{align}
\end{proposition}

\begin{proof}
	We only address the $Accr$ case, as the $Depl$ one is similar (actually, the latter is simpler).
	By the proof of Proposition~\ref{prop:aggregDeplAndAccr} we obtain that we can rewrite the left-hand side of Equation~\eqref{eq:aggrAccrRedAccrInProp} as
	\begin{align}
	&\sum_{\tilde \Y \in \pa}\gr[Y^{\tilde \Y},\emptyset,\Y]\cdot  V_{\tilde \Y}\ +\label{eq:aggrAccrRedAccrInPropLHS1}\\
	&\quad\frac{1}{2}\sum_{\tilde \Y \in \pa} V_{\tilde \Y}\sum_{\widehat \Y \in \pa} V_{\widehat \Y}\cdot\gr[Y^{\widehat{\Y}},{Y^{\tilde \Y}},\Y]\label{eq:aggrAccrRedAccrInPropLHS2}
	\end{align}
	%
	Instead, we can rewrite the right-hand side of Equation~\eqref{eq:aggrAccrRedAccrInProp} as
	\begin{align}
	&
	\sum_{\stackrel{Y \act{\rate} \pi \in R}{Y=\mu(Y)}} \!\!\! \mu_\pi(\canonical{\Y}) \cdot \rate \cdot   V_Y^{(\pa,F)}
	\ +\label{eq:aggrAccrRedAccrInPropRHS1}\\
	&\quad
	\sum_{\stackrel{Y_1+Y_2 \act{\rate} \pi \in R}{Y_1=\mu(Y_1) \wedge Y_2=\mu(Y_2)}} \!\!\! \mu_\pi(\canonical{\Y}) \cdot \rate \cdot   V_{Y_1}^{(\pa,F)}\cdot V_{Y_2}^{(\pa,F)}
	\label{eq:aggrAccrRedAccrInPropRHS2}
	\end{align}
	We close the proof showing  that Equation~\eqref{eq:aggrAccrRedAccrInPropLHS1} is equal to  Equation~\eqref{eq:aggrAccrRedAccrInPropRHS1} and that Equation~\eqref{eq:aggrAccrRedAccrInPropLHS2} is equal to Equation~\eqref{eq:aggrAccrRedAccrInPropRHS2}.
	
	We focus on the first equality, which easily follows noting that  Equation~\eqref{eq:aggrAccrRedAccrInPropRHS1} can be rewritten as
	\begin{align*}
	\sum_{\tilde\Y\in\pa} V_{\canonical{\tilde\Y}}^{(\pa,F)}\!\!\!\sum_{\stackrel{Y \act{\rate} \pi \in R}{Y=\canonical{\tilde\Y}}} \!\!\! \mu_\pi(\canonical{\Y}) \cdot \rate
	=
	\sum_{\tilde\Y\in\pa}  V_{\tilde\Y}
	\!\!
	\sum_{Z\in\Y}
	\sum_{\stackrel{Y \act{\rate} \pi \in R}{Y=\canonical{\tilde\Y}}}\!\!
	\pi(Z) \cdot \rate
	=
	\sum_{\tilde\Y\in\pa}  V_{\tilde\Y} \cdot \gr[\canonical{\tilde\Y},\emptyset,\Y]
	\end{align*}

	We now consider the case ``Equation~\eqref{eq:aggrAccrRedAccrInPropLHS2} equals Equation~\eqref{eq:aggrAccrRedAccrInPropRHS2}''.
	We can rewrite Equation~\eqref{eq:aggrAccrRedAccrInPropRHS2} as follows:
	\begin{align}
	&\sum_{\stackrel{Y_1+Y_1 \act{\rate} \pi \in R}{Y_1=\mu(Y_1)}} \!\!\!\!\!\! \mu_\pi(\canonical{\Y}) \cdot \rate \cdot   V_{Y_1}^{(\pa,F)}\cdot V_{Y_1}^{(\pa,F)} +\label{eq:aggrAccrRedAccrInPropRHS2Y1Y1}\\
	&\quad\sum_{\stackrel{Y_1+Y_2 \act{\rate} \pi \in R, Y_1\not=Y_2}{Y_1=\mu(Y_1) \wedge Y_2=\mu(Y_2)}} \!\!\!\!\!\!\!\!\!\! \mu_\pi(\canonical{\Y}) \cdot \rate \cdot   V_{Y_1}^{(\pa,F)}\cdot V_{Y_2}^{(\pa,F)}\label{eq:aggrAccrRedAccrInPropRHS2Y1Y2}
	%
	%
	\end{align}
	Now, Equation~\eqref{eq:aggrAccrRedAccrInPropRHS2Y1Y1} can be rewritten as
	\begin{align}
	&\sum_{\tilde \Y\in\pa} V_{\tilde\Y}\cdot  V_{\tilde\Y} \sum_{\stackrel{Y_1+Y_1 \act{\rate} \pi \in R}{Y_1=\canonical{\tilde\Y}}} \!\!\!\!\!\! \mu_\pi(\canonical{\Y}) \cdot \rate=\notag\\
	&\sum_{\tilde \Y\in\pa} V_{\tilde\Y}\cdot  V_{\tilde\Y}\sum_{Z\in\Y} \sum_{\stackrel{Y_1+Y_1 \act{\rate} \pi \in R}{Y_1=\canonical{\tilde\Y}}} \!\!\!\!\!\!\pi(Z) \cdot \rate= \comInAlign{multiplying by $\frac{2}{2}$}\notag
	\\
	&
	\frac{1}{2}\sum_{\tilde \Y\in\pa} V_{\tilde\Y}\cdot  V_{\tilde\Y}\sum_{Z\in\Y} 2\sum_{\stackrel{Y_1+Y_1 \act{\rate} \pi \in R}{Y_1=\canonical{\tilde\Y}}} \!\!\!\!\!\!\pi(Z) \cdot \rate=\notag\\
	&
	\frac{1}{2}\sum_{\tilde \Y\in\pa} V_{\tilde\Y}\cdot  V_{\tilde\Y}\cdot\gr[\canonical{\tilde\Y},\canonical{\tilde\Y},\Y]\label{eq:aggrAccrRedAccrInPropRHS2Y1Y1-1}
	\end{align}
	Instead,
	Equation~\eqref{eq:aggrAccrRedAccrInPropRHS2Y1Y2} can be rewritten as, where we divide Equation~\eqref{eq:aggrAccrRedAccrInPropRHS2Y1Y2-1} by $2$ because each multiset is considered twice (e.g., $Y_1+Y_2$ and $Y_2+Y_1$):
	\begin{align}
	%
	%
	&\frac{1}{2}\sum_{\tilde \Y \in \pa} V_{\tilde\Y}\sum_{\widehat \Y \in \pa, \tilde \Y\not=\widehat\Y} V_{\widehat \Y}
	\sum_{\stackrel{Y_1+Y_2 \act{\rate} \pi \in R}{Y_1=\canonical{\tilde \Y}\wedge Y_2=\canonical{\widehat \Y}}} \!\!\!\!\!\! \mu_\pi(\canonical{\Y}) \cdot \rate=\label{eq:aggrAccrRedAccrInPropRHS2Y1Y2-1}\\
	&\frac{1}{2}\sum_{\tilde \Y \in \pa} V_{\tilde\Y}\sum_{\widehat \Y \in \pa, \tilde \Y\not=\widehat\Y} V_{\widehat \Y}
	\sum_{Z\in\Y}
	\sum_{\stackrel{Y_1+Y_2 \act{\rate} \pi \in R}{Y_1=\canonical{\tilde \Y}\wedge Y_2=\canonical{\widehat \Y}}} \!\!\!\!\!\!
	\pi(Z) \cdot \rate=\notag\\
	&\frac{1}{2}\sum_{\tilde \Y \in \pa} V_{\tilde\Y}\sum_{\widehat \Y \in \pa, \tilde \Y\not=\widehat\Y} V_{\widehat \Y}
	\cdot\gr[\canonical{\tilde \Y},\canonical{\widehat \Y},\Y]\label{eq:aggrAccrRedAccrInPropRHS2Y1Y2-2}
	\end{align}
	By summing  Equation~\eqref{eq:aggrAccrRedAccrInPropRHS2Y1Y1-1} and Equation~\eqref{eq:aggrAccrRedAccrInPropRHS2Y1Y2-2}, we finally rewrite Equation~\eqref{eq:aggrAccrRedAccrInPropRHS2} as
	\begin{align*}
	&\frac{1}{2}\sum_{\tilde Y\in\pa} V_{\tilde\Y}\cdot V_{\tilde\Y}\cdot\gr[\canonical{\tilde\Y},\canonical{\tilde\Y},\Y]+
	\frac{1}{2}\sum_{\tilde \Y \in \pa} V_{\tilde\Y}\sum_{\widehat \Y \in \pa, \tilde \Y\not=\widehat\Y} V_{\widehat \Y}
	\cdot\gr[\canonical{\tilde \Y},\canonical{\widehat \Y},\Y]=\\
	&\frac{1}{2}\sum_{\tilde \Y \in \pa} V_{\tilde\Y}\sum_{\widehat \Y \in \pa} V_{\widehat \Y}
	\cdot\gr[\canonical{\tilde \Y},\canonical{\widehat \Y},\Y]
	\end{align*}
	This closes the case, and thus the proof is complete.
\end{proof}

We now provide the proof of Theorem~\ref{th:reducedModelHasLumpedODEs}.

\begin{proof}[Theorem~\ref{th:reducedModelHasLumpedODEs}]
	We first address the correctness of the reduction.
	For any $\Y\in\pa$ we have
	\[
	\sum_{X\in\Y}F_X( V) = \sum_{X\in\Y}\sum_{\rho \act{\rate} \pi \in R} Accr(\rho \act{\rate} \pi,X, V) - \sum_{X\in\Y}\sum_{\rho \act{\rate} \pi \in R} Depl(\rho \act{\rate} \pi,X, V)
	\]
	As regards $\hat{F}$, by 
	Definition~\ref{def:OReduction}, for any $\Y \in \pa$ we obtain
	\begin{align*}
	\hat{F}_{\canonical{\Y}}(\reduce{V}{(\pa,F)}) = &\sum_{\stackrel{\rho \act{\rate} \pi \in R}{\rho=\mu(\rho)}} Accr(\rho \act{\rate} \mu(\pi),\canonical{\Y},\reduce{V}{(\pa,F)}) -
	\\&\quad
	\sum_{\stackrel{\rho \act{\rate} \pi \in R}{\rho=\mu(\rho)}} Depl(\rho \act{\rate} \mu(\pi),\canonical{\Y},\reduce{V}{(\pa,F)})
	\end{align*}
	We close the proof by showing that
	\begin{align}
	\sum_{X\in\Y}\sum_{\rho \act{\rate} \pi \in R} Accr(\rho \act{\rate} \pi,X, V) = \sum_{\stackrel{\rho \act{\rate} \pi \in R}{\rho=\mu(\rho)}} Accr(\rho \act{\rate} \mu(\pi),\canonical{\Y},\reduce{V}{(\pa,F)})\label{eq:aggrAccrRedAccr}\\
	\sum_{X\in\Y}\sum_{\rho \act{\rate} \pi \in R} Depl(\rho \act{\rate} \pi,X, V) = \sum_{\stackrel{\rho \act{\rate} \pi \in R}{\rho=\mu(\rho)}} Depl(\rho \act{\rate} \mu(\pi),\canonical{\Y},\reduce{V}{(\pa,F)})
	\label{eq:aggrDeplRedDepl}
	\end{align}
	%
	Both Equations~\eqref{eq:aggrAccrRedAccr} and~\eqref{eq:aggrDeplRedDepl} follow from Proposition~\ref{prop:aggrAccrRedAccr}.
	
	\medskip
	
	We now address the complexity of the reduction, showing that $\reduce{(S,R)}{(\pa,F)}$ can be performed in $\mathcal{O} \big( |R| \cdot |S| \cdot (\log(\abs{S})+\log(\abs{R}))\big)$ time.
	%
	
	Steps (F1) and (F2) of Definition~\ref{def:OReduction} require to iterate (once) the reactions, ($\mathcal{O}(|R|)$). In particular, for each reaction we have to in turn iterate its reagents and products to perform (F1) and (F2), respectively. This requires $\mathcal{O}(\abs{S})$ time. Finally, in order to efficiently perform (F3) we assume that the reagents and products are stored as an ordered list, and thus we have to sort the obtained canonized products ($\mathcal{O}(\abs{S}\cdot \log(\abs{S}))$).
	To sum up, steps (F1) and (F2) require $\mathcal{O} \big( |R| \cdot |S| \cdot \log(\abs{S})\big)$ time.
	
	Instead, step (F3) can be computed by first sorting the reactions obtained from (F1) and (F2), requiring $\mathcal{O}(\abs{R}\cdot \log(\abs{R})\cdot \abs{S})$, where the $\abs{S}$ factor comes from the fact that in order to compare two reactions it is necessary to scan (once) their reagents and products.
	Then, (F3) is  completed by iterating (once) the reactions to actually collapse them ($\mathcal{O}(|R|\cdot\abs{S})$).
\end{proof}

\subsection*{A.2\ Backward CRN Bisimulation}

\begin{proof}[Proof of Theorem \ref{thm_efl}]
	We first prove the if direction. Let $\mu : S \to S$ be some choice function of $\mathcal{H}$, set $\hat{S} := \mu(S)$ and define $G_{\hat{X}}(\hat{{V}}) := F_{\hat{X}}(\hat{{V}} \circ \mu)$ for any $\hat{{V}} \in \mathbb{R}^{\hat{S}}$ and $\hat{X} \in \hat{S}$. Further, let $\hat{{V}}$ denote the unique ODE solution of $\frac{d}{dt}{\hat{{V}}}(t) = G\big(\hat{{V}}(t)\big)$ subject to some given initial condition $\hat{{V}}(0)$. Then, for all $X \in S$, it holds that
	\begin{multline*}
	\frac{d}{dt}(\hat{{V}}(t) \circ \mu)_X = \Big(\frac{d}{dt}\hat{{V}}(t)\Big)_{\mu(X)}
	= G_{\mu(X)}\big(\hat{{V}}(t)\big) = F_{\mu(X)}(\hat{{V}}(t) \circ \mu) = F_{X}(\hat{{V}}(t) \circ \mu)
	\end{multline*}
	Thus, $t \mapsto \hat{{V}}(t) \circ \mu$ is the unique solution of the ODE system $\frac{d}{dt} {V}(t) = F({V}(t))$ subject to $\hat{{V}}(0) \circ \mu$. Since $t \mapsto \hat{{V}}(t) \circ \mu$ is constant on $\mathcal{H}$, the proof of the if direction is complete. We now turn to the proof of the only-if direction. For this, let us assume that $F$ is such that, for any $V(0) \in \RE_{\geq0}^S$ that is constant on $\pa$, the underlying solution of $\dot{V} = F(V)$ is constant on $\pa$ as well. Fix arbitrary $H \in \pa$, $X,Y \in H$ and $V(0) \in \RE_{\geq0}^S$ with $V(0)$ constant on $\pa$. Since the solution of $\dot{V} = F(V)$ is differentiable, there exist $\delta > 0$ and functions $r_X,r_Y$ such that $\lim_{h \searrow 0} \frac{1}{h} r_X(h) = \lim_{h \searrow 0} \frac{1}{h} r_Y(h) = 0$ and
\begin{align*}
V_X(h) & = V_X(0) + F_X(V(0)) \cdot h + r_X(h) &
V_Y(h) & = V_Y(0) + F_Y(V(0)) \cdot h + r_Y(h)
\end{align*}
for all $0 \leq h < \delta$. Since $V(0)$ is constant on $\pa$ and $X,Y \in H$ for some $H \in \pa$, this yields $\lim_{h \searrow 0} \frac{1}{h}(V_X(h) - V_Y(h)) = F_X(V(0)) - F_Y(V(0))$. Noting that $V_X(h) = V_Y(h)$ for all $0 \leq h < \delta$ because $V(0)$ is constant on $\pa$, we thus infer $F_X(V(0)) - F_Y(V(0)) = 0$.
\end{proof}


\begin{proof}[Proof of Proposition~\ref{prop:cefsb:ex}]
	We first note that $\R$ is an equivalence relation over $S$, as it is the transitive closure of the union of equivalence relations over $S$. For $i \in I$, set $\pa_i = S/ \R_i$ and $\pa = S / \R$. Note that, for all $i \in I$, any $\Y \in \pa$ is a union of blocks of $\pa_i$. From this, in turn, it is easy to see that any $\mathcal{M} \in \{ \rho \mid \rho \act{\alpha} \pi \in R \} /\!\!\! \approx_\pa$ can be written as a union of blocks of $\{ \rho \mid \rho \act{\alpha} \pi \in R \} /\!\!\!\approx_{\pa_i}$. Observe that, for any $(y_0,y_1) \in \R$, there exist $x_0,\ldots,x_k \in S$ such that $x_0 R_{i_0} x_1 R_{i_1} \ldots R_{i_{k-1}} x_k$ with $y_0 = x_0$, $y_1 = x_k$ and $i_j \in I$ for all $0 \leq j \leq k - 1$. Noting that $x_j R_{i_j} x_{j + 1}$ for all $0 \leq j \leq k - 1$, we infer that $\fr[x_j,\M] = \fr[x_{j+1},\M]$ for all $\M \in \{ \rho \mid \rho \act{\alpha} \pi \in R \} /\!\!\!\approx_{\pa_{i_j}}$. Thus, we infer that $\fr[x_0,\M] = \fr[x_1,\M] = \ldots = \fr[x_{k - 1},\M] = \fr[x_k,\M]$ for all $\M \in \{ \rho \mid \rho \act{\alpha} \pi \in R \} /\!\!\approx_{\pa}$. The remainder of the claim follows then from Lemma~\ref{lem_ex_cor_aux}.
\end{proof}

\begin{proof}[Proof of Theorem~\ref{th:efsbAndODELumping}]
	Define $\llbracket \rho \rrbracket_V := \prod_{X \in S} V_X^{\rho(X)}$ and set $Q := \{ \rho \mid \rho \act{\alpha} \pi \in R \} /\!\!\!\approx_\pa$. Fix some arbitrary $X_i,X_j \in H$ and $H \in \pa$ and note that
	\begin{align*}
	F_{X_k}(V) & = \sum_{\rho \act{\alpha} \pi \in R} \alpha (\pi(X_k) - \rho(X_k)) \llbracket \rho \rrbracket_V = \sum_{[\rho_0] \in Q} \sum_{\rho \in [\rho_0]} \fr(X_k,\rho) \llbracket \rho \rrbracket_V \\
	& = \sum_{[\rho_0] \in Q} \underbrace{\Big( \sum_{\rho \in [\rho_0]} \fr(X_k,\rho) \Big)}_{c(X_k,[\rho_0])} \llbracket \rho_0 \rrbracket_V
	\end{align*}
	whenever $V$ is constant on $\pa$. Observe also that the function $V \mapsto F_{X_k}(V)$, where $V$ is constant on $\pa$, defines a polynomial in $|Q|$ variables with the monomials $\{ c(X_k,[\rho_0]) \cdot \llbracket \rho_0 \rrbracket_V \mid [\rho_0] \in Q\}$. At last, recall that the multi dimensional version of Taylor's theorem implies that two real polynomials are equivalent if and only if they have the same monomials. Thus, $F_{X_i}(V) = F_{X_j}(V)$ for all $V$ that are constant on $\pa$ if and only if $c(X_i,[\rho_0]) = c(X_j,[\rho_0])$ for all $[\rho_0] \in Q$.
\end{proof}

\begin{proof}[Proof of Theorem \ref{th:e:reduction}]
	By encoding $\mu : S \to S$, $\rho$ and $\pi$ as arrays of length $|S|$, it is easy to see that the first operation needs at most $\mathcal{O}(|R| \cdot |S|)$ steps. For the second operation, note that the renaming of species according to $\mu$ can be done in again in $\mathcal{O}(|R| \cdot |S|)$. However, since elements of $\ms(S)$ are stored as ordered lists to allow for performant processing, the second operation needs $\mathcal{O}\big(|R| \cdot |S| \cdot \log(|S|)\big)$. To accomplish the third operation, instead, we first sort the reactions with respect to the lexicographical order which takes $\mathcal{O}(|R| \cdot \log(|R|) \cdot |S|)$. Afterwards, the rates of the reactions that coincide in reactants and products can be summed in $\mathcal{O}(|R| \cdot |S|)$.
	
	We now turn to the correctness of algorithm. Let $(X,{V}) \mapsto G^i_X({V})$ denote the vector field that arises from $R$ after applying (B$1$), \ldots, (B$i$), with $G^0$ being $F$ itself. We next prove that $G^i_{X_k}(\hat{{V}} \circ \mu_S) = F_{X_k}(\hat{{V}} \circ \mu_S)$ for all $\hat{{V}} \in \mathbb{R}^{\hat{S}}_{\geq0}$, $X_k \in \hat{S}$ and $i \in \{1,2\}$. For this, let us first apply the reaction changes $\rho \to^\alpha \pi \mapsto \rho \to^\alpha \tilde{\pi}$ of (B1). Then, $G^0_{X_k}(\hat{{V}} \circ \mu) = G^1_{X_k}(\hat{{V}} \circ \mu)$ because of $\pi(X_k) - \rho(X_k) = \tilde{\pi}(X_k) - \rho(X_k)$. Let us now consider a reaction changes $\rho \to^\alpha \pi \mapsto \tilde{\rho} \to^\alpha \tilde{\pi}$ of (B2). Then, $G^1_{X_k}(\hat{{V}} \circ \mu) = G^2_{X_k}(\hat{{V}} \circ \mu)$ because $\prod_{X \in S} (\hat{{V}} \circ \mu)_X^{\rho(X)} = \prod_{X \in S} (\hat{{V}} \circ \mu)_X^{\tilde{\rho}(X)}$ and $\pi(X_k) - \rho(X_k) = \tilde{\pi}(X_k) - \tilde{\rho}(X_k)$. Since $G^2_{X_k}(\hat{{V}} \circ \mu) = G^3_{X_k}(\hat{{V}} \circ \mu)$ is trivially true, we infer the claim.
\end{proof}

\subsection*{A.3\ Partition Refinement}

The following auxiliary results will be needed to prove the correctness of Algorithm~\ref{algorithm_cfsb}.

\begin{lemma}\label{lem_key}
	Given a CRN $(S,R)$, let $\mathcal{H}_1, \mathcal{H}_2$ be two partitions of $S$ such that $\mathcal{H}_1$ is a refinement of $\mathcal{H}_2$. Then, the following holds true.
	\begin{enumerate}
		\item $X_i \sim_{\mathcal{H}_1}^\OSB{} X_j$ implies $X_i \sim_{\mathcal{H}_2}^\OSB{} X_j$.
		\item $X_i \sim_{\mathcal{H}_1}^\ESB{} X_j$ implies $X_i \sim_{\mathcal{H}_2}^\ESB{} X_j$.
	\end{enumerate}
\end{lemma}

\begin{proof}
	Let us assume that $X_i \sim_{\mathcal{H}_1}^\OSB{} X_j$, i.e. it holds that $\crr[X_i,\rho]= \crr[X_j,\rho]$ and $\gr[X_i,\rho,H_1] = \gr[X_j,\rho,H_1]$ for all $\rho \in \ms(S)$ and $H_1 \in \pa_1$. Since any $H_2 \in \pa_2$ can be written as a disjoint union of blocks from $\pa_1$, we thus infer that $X_i \sim_{\mathcal{H}_2}^\OSB{} X_j$. Let us now assume that $X_i \sim_{\mathcal{H}_1}^\ESB{} X_j$, i.e. it holds that $\fr[X_i,\M_1] = \fr[X_j,\M_1]$ for all $\M_1 \in \{ \rho \mid \rho \act{\alpha} \pi \in R \} /\!\!\approx_{\pa_1}$. Since any $\M_2 \in \{ \rho \mid \rho \act{\alpha} \pi \in R \} /\!\! \approx_{\pa_2}$ can be written as a disjoint union of blocks from $\{ \rho \mid \rho \act{\alpha} \pi \in R \} /\!\!\approx_{\pa_1}$, we deduce that $X_i \sim_{\mathcal{H}_2}^\ESB{} X_j$.
\end{proof}

\begin{lemma}\label{lem_key_2}
	Let $(S,R)$ be a CRN and $\pa$ a partition of $S$. Then, the following holds.
	\begin{enumerate}
		\item $\pa$ is an \FB{} if and only if $X_i \sim_\pa^\OSB{} X_j$ for all $X_i,X_j \in H$ and $H \in \pa$. Moreover, it holds that $\pa$ is an \FB{} if and only if $\pa = S / (\sim_\pa^\OSB{} \cap \sim_\pa)$.
		\item $\pa$ is a \BB{} if and only if $X_i \sim_\pa^\ESB{} X_j$ for all $X_i,X_j \in H$ and $H \in \pa$. Moreover, it holds that $\pa$ is a \BB{} if and only if $\pa = S / (\sim_\pa^\ESB{} \cap \sim_\pa)$.
	\end{enumerate}
\end{lemma}

\begin{proof}
	The first parts of 1. and 2. are straightforward. The second parts, instead, follow from the corresponding first parts by observing that $\mathcal{H} = S / (\sim^\chi_\mathcal{H} \cap \sim_\mathcal{H})$ if and only if $\mathcal{H}$ is a refinement of $S / \sim^\chi_\mathcal{H}$.
\end{proof}

The following auxiliary results will be needed to establish polynomial complexity of Algorithm~\ref{algorithm_cfsb}. We implement $\rho \in \ms(S)$ as maps with keys in $S$. Moreover, we store a partition of a set $A$ by means of a map that associates to each element of $A$ a pointer to its representative. That is, each element of a partition block has a pointer to the representative of the block.

\begin{algorithm}[t!]
	\caption{Algorithm to calculate the quotient $A / \sim$. We assume that $A$ is implemented as an array of objects $\{a[1],\ldots,a[n]\}$ where each object contains, among the actual data, a pointer that is initialized with zero at the beginning.}\label{algorithm_comp_quot}
	\begin{algorithmic}
		\REQUIRE A set $A$ and an equivalence relation $\sim$ on $A$.
		\FOR{$i = 1$ to $n$}
		\IF{$a[i].p \text{ != \textbf{null}}$}
		\STATE \text{\textbf{continue}}
		\ENDIF
		
		\STATE $a[i].p \leftarrow \& a[i]$
		
		\FOR{$j = i + 1$ to $n$}
		\IF{$a[j].p \text{ == \textbf{null} \&\& } a[j] \sim a[i]$}
		\STATE $a[j].p \leftarrow \& a[i]$
		\ENDIF
		\ENDFOR
		\ENDFOR
	\end{algorithmic}
\end{algorithm}

\begin{lemma}\label{lem_comp_quot}
	Fix a CRN $(S,R)$, pick $A \in \{S,R\}$ and assume that deciding $a_1 \sim a_2$ for some equivalence relation $\sim$ on $A$ can be done in $\mathcal{O}(|R|^{e_1} \cdot |S|^{e_2})$ steps. Then, $A / \sim$ can be calculated in $\mathcal{O}(|A|^2 \cdot |R|^{e_1} \cdot |S|^{e_2})$ steps.
\end{lemma}

\begin{proof}
	It can be easily seen that Algorithm~\ref{algorithm_comp_quot} calculates $A / \sim$ in $\mathcal{O}(|A|^2 \cdot |R|^{e_1} \cdot |S|^{e_2})$.
\end{proof}

\begin{lemma}\label{lem_comp_o}
	For a CRN $(S,R)$ and a partition $\pa$ of $S$, deciding $X \sim_\pa^\OSB{} Y$ can be done in $\mathcal{O}(|R|^2 \cdot |S|^2)$.
\end{lemma}

\begin{proof}
	We first note that, for a given $\rho \in \ms(S)$, deciding $\crr[X,\rho] = \crr[Y,\rho]$ can be done in $\mathcal{O}(|R| \cdot |S|)$ because the comparison of two $\rho_1,\rho_2 \in \ms(S)$ needs $\mathcal{O}(|S|)$. Similarly, for given $\rho \in \ms(S)$ and $Z \in S$, the calculation of $\gr[X,\rho,Z]$ can be done in $\mathcal{O}(|R| \cdot |S|)$. Note also that $\mathcal{D}(Z) := \{ \rho \mid \exists Z + \rho \act{\alpha} \pi \in R \}$ can be calculated in $ \mathcal{O}(|R| \cdot |S|)$ steps. We are now in a position to infer the claim. For this, note that $\crr[X,\rho] = \crr[Y,\rho]$ for all $\rho \in \ms(S)$ if $\crr[X,\rho] = \crr[Y,\rho]$ for all $\rho \in \mathcal{D}(X) \cup \mathcal{D}(Y)$. Consequently, deciding whether $\crr[X,\rho] = \crr[Y,\rho]$ for all $\rho \in \ms(S)$ can be done in $\mathcal{O}(|R|^2 \cdot |S|)$ steps. Similarly, we note that $\gr[X,\rho,H] = \gr[Y,\rho,H]$ for all $\rho \in \ms(S)$ and $H \in \pa$ if $\gr[X,\rho,H] = \gr[Y,\rho,H]$ for all $\rho \in \mathcal{D}(X) \cup \mathcal{D}(Y)$ and $H \in \pa$. Thus, deciding whether $\gr[X,\rho,H] = \gr[Y,\rho,H]$ for all $\rho \in \ms(S)$ and $H \in \pa$ holds true can be done in $\mathcal{O}(|R|^2 \cdot |S|^2)$.
\end{proof}

\begin{lemma}\label{lem_comp_e}
	For a CRN $(S,R)$ and a partition $\pa$ of $S$, deciding $X \sim_\pa^\ESB{} Y$ can be done in $\mathcal{O}(|R|^2 \cdot |S|)$.
\end{lemma}

\begin{proof}
	Note that we can decide whether $\rho \approx_\pa \sigma$ in $\mathcal{O}(|S|)$. Thus, since $\mathcal{R} = \{ \rho \mid \rho \act{\alpha} \pi \in R\}$ can be calculated in $\mathcal{O}(|R| \cdot |S|)$ steps, Lemma~\ref{lem_comp_quot} implies that $Q = \mathcal{R} / \approx_\pa$ can be calculated in $\mathcal{O}(|R|^2 \cdot |S|)$. Moreover, we note that, for a given $\rho \in \ms(S)$, deciding whether $\fr[X,\rho] = \fr[Y,\rho]$ can be done in $\mathcal{O}(|R| \cdot |S|)$. Consequently, deciding whether $X \sim_\pa^\ESB{} Y$ can be done in $\mathcal{O}(|R|^2 \cdot |S|)$.
\end{proof}

We now are in a position to prove Theorem~\ref{thm_cfsb}.

\begin{proof}[Proof of Theorem \ref{thm_cfsb}]
	For the proof of correctness, let us assume that $\mathcal{H}$ denotes the coarsest bisimulation that refines $\mathcal{H}_0 := \mathcal{G}$ and define $\mathcal{H}_{k + 1} := S / (\sim^\chi_{\mathcal{H}_k} \cap \sim_{\mathcal{H}_k})$ for all $k \geq 0$. Then, the sequence $\mathcal{H}_0, \mathcal{H}_1, \mathcal{H}_2, \ldots$ is such that
	\begin{enumerate}
		\item $\mathcal{H}$ is a refinement of $\mathcal{H}_k$
		\item $\mathcal{H}_k$ is a refinement of $\mathcal{H}_{k - 1}$
	\end{enumerate}
	for all $k \geq 1$. We prove this by induction on $k$.
	\begin{itemize}
		\item $k = 1$: Since $\mathcal{H}$ is a refinement of $\mathcal{H}_0$, Lemma~\ref{lem_key} ensures the first claim. The second claim is trivial.
		\item $k \to k + 1:$ Thanks to the fact that $\mathcal{H}$ is a refinement of $\mathcal{H}_k$ by induction, Lemma~\ref{lem_key} ensures the first claim. The second claim is trivial.
	\end{itemize}
	Since $\mathcal{H}$ is a refinement of any $\mathcal{H}_k$, it holds that $\mathcal{H} = \mathcal{H}_k$ whenever $\mathcal{H}_k$ is a bisimulation. Thanks to the fact that $\mathcal{H}_k$ is a refinement of $\mathcal{H}_{k - 1}$ for all $k \geq 1$ and $S$ is finite, we can fix the smallest $k \geq 1$ such that $\mathcal{H}_k = \mathcal{H}_{k - 1}$. Since this implies $\mathcal{H}_{k - 1} = \mathcal{H}_k = S / (\sim^\chi_{\mathcal{H}_{k - 1}} \cap \sim_{\mathcal{H}_{k-1}})$, Lemma~\ref{lem_key_2} yields the claim. We now turn to the complexity analysis. Note that since deciding $X \sim_\pa Y$ needs a constant number of steps, Lemma~\ref{lem_comp_o} and Lemma~\ref{lem_comp_e} imply that $X \sim^\chi_\pa Y$ can be decided in $\mathcal{O}(|R|^2 \cdot |S|^2)$ steps. Thus, Lemma~\ref{lem_comp_quot} ensures that $S / \sim^\chi_\pa$ can be calculated in $\mathcal{O}(|R|^2 \cdot |S|^4)$ steps. Noting that the algorithm makes at most $|S|$ iterations, we conclude that the algorithm needs at most $\mathcal{O}(|R|^2 \cdot |S|^5)$ steps.
\end{proof}

\subsection*{A.4\ $\kappa$-encodings discussed in Section~\ref{sec:caseStudies}}


%


\paragraph*{M1-M4}
We start providing the $\kappa$-encoding of model M4 of \autoref{table:reductionResults}, whose original (BNGL) version has been taken from~\cite{citeulike:8493139}.
M4 is a special case in which almost no changes are necessary to convert it in $\kappa$.
The encoding is given in \autoref{ls:kappaEncodingM8}, where we omit unnecessary details. Each $\kappa$ rule is preceded by the corresponding original BNGL rule (starting with $\#$).
The encodings for all other models of the same benchmark, M1-M3, are similar and thus we omit them here. They are available for download at \url{http://sysma.imtlucca.it/crnreducer/}. 

\lstset{basicstyle=\ttfamily\scriptsize,breaklines=true,caption=Encoding of M4 from BNGL into $\kappa$,  label=ls:kappaEncodingM8,numbers=left,morecomment=[l]{\#}}
\begin{lstlisting}[mathescape,float=t]
###  SITE 1 ###
#E(s) + S(p1~U) <-> E(s!1).S(p1~U!1)
E(s) , S(p1~U) <-> E(s!1),S(p1~U!1)

#E(s!1).S(p1~U!1) -> E(s) + S(p1~P)
E(s!1),S(p1~U!1) -> E(s) , S(p1~P)

#F(s) + S(p1~P) <-> F(s!1).S(p1~P!1)
F(s) , S(p1~P) <-> F(s!1),S(p1~P!1)

#F(s!1).S(p1~P!1) -> F(s) + S(p1~U)
F(s!1),S(p1~P!1) -> F(s) , S(p1~U)


###  SITE 2 ###
#E(s) + S(p2~U) <-> E(s!1).S(p2~U!1)
E(s) , S(p2~U) <-> E(s!1),S(p2~U!1)

#E(s!1).S(p2~U!1) -> E(s) + S(p2~P)
E(s!1),S(p2~U!1) -> E(s) , S(p2~P)

#F(s) + S(p2~P) <-> F(s!1).S(p2~P!1)
F(s) , S(p2~P) <-> F(s!1),S(p2~P!1)

#F(s!1).S(p2~P!1) -> F(s) + S(p2~U)
F(s!1),S(p2~P!1) -> F(s) , S(p2~U)
\end{lstlisting}

\paragraph*{M6-M7}
As discussed in Section~\ref{sec:caseStudies}, models M6 and M7 are not directly encodable in $\kappa$, as they contain the molecule $Lig(l,l)$ having two identical binding sites $l$, which is forbidden in $\kappa$.
Hence, we encoded in $\kappa$ an expansion of the models similar to that of Equation~\eqref{eq:expansionIdenticalSite}.
For presentation reasons, we do not provide the $\kappa$-encodings of the expansions of M6 and M7, which however are available for download at \url{http://sysma.imtlucca.it/crnreducer/}.

\paragraph*{M12}
We conclude this appendix providing in \autoref{ls:kappaEncodingM22} the $\kappa$-encoding of M12, whose BNGL version has been taken from~\cite{Kocieniewski2012116}. As for M4, each $\kappa$ rule is preceded by the corresponding original BNGL rule.

\lstset{caption=Encoding of M12 from BNGL to $\kappa$,  label=ls:kappaEncodingM22,numbers=left,morecomment=[l]{\#}}
\begin{lstlisting}[mathescape,float=t]
#MAP3K(s,S~I) -> MAP3K(s,S~A)
MAP3K(s,S~I) -> MAP3K(s,S~A)

#MAP3K(s,S~A)+Scaff(map3k) <-> MAP3K(s!1,S~A).Scaff(map3k!1)
MAP3K(s,S~A),Scaff(map3k) <-> MAP3K(s!1,S~A),Scaff(map3k!1)

#MAP3K(s!1,S~I).Scaff(map3k!1) -> MAP3K(s,S~I)+Scaff(map3k)
MAP3K(s!1,S~I),Scaff(map3k!1) -> MAP3K(s,S~I),Scaff(map3k)

#MAP2K(s,R1~Y,R2~Y)+Scaff(map2k) <-> MAP2K(s!1,R1~Y,R2~Y).Scaff(map2k!1)
MAP2K(s,R1~Y,R2~Y),Scaff(map2k) <-> MAP2K(s!1,R1~Y,R2~Y),Scaff(map2k!1)

#MAP2K(s,R1~Yp,R2~Y)+Scaff(map2k) <-> MAP2K(s!1,R1~Yp,R2~Y).Scaff(map2k!1)
MAP2K(s,R1~Yp,R2~Y),Scaff(map2k) <-> MAP2K(s!1,R1~Yp,R2~Y),Scaff(map2k!1)

#MAP2K(s,R1~Y,R2~Yp)+Scaff(map2k) <-> MAP2K(s!1,R1~Y,R2~Yp).Scaff(map2k!1)
MAP2K(s,R1~Y,R2~Yp),Scaff(map2k) <-> MAP2K(s!1,R1~Y,R2~Yp),Scaff(map2k!1)

#MAP2K(s,R1~Yp,R2~Yp)+Scaff(map2k) <-> MAP2K(s!1,R1~Yp,R2~Yp).Scaff(map2k!1)
MAP2K(s,R1~Yp,R2~Yp),Scaff(map2k) <-> MAP2K(s!1,R1~Yp,R2~Yp),Scaff(map2k!1)

#MAPK(s,R1~Y,R2~Y)+Scaff(mapk) <-> MAPK(s!1,R1~Y,R2~Y).Scaff(mapk!1)
MAPK(s,R1~Y,R2~Y),Scaff(mapk) <-> MAPK(s!1,R1~Y,R2~Y),Scaff(mapk!1)

#MAPK(s,R1~Yp,R2~Y)+Scaff(mapk) <-> MAPK(s!1,R1~Yp,R2~Y).Scaff(mapk!1)
MAPK(s,R1~Yp,R2~Y),Scaff(mapk) <-> MAPK(s!1,R1~Yp,R2~Y),Scaff(mapk!1)

#MAPK(s,R1~Y,R2~Yp)+Scaff(mapk) <-> MAPK(s!1,R1~Y,R2~Yp).Scaff(mapk!1)
MAPK(s,R1~Y,R2~Yp),Scaff(mapk) <-> MAPK(s!1,R1~Y,R2~Yp),Scaff(mapk!1)

#MAPK(s!1,R1~Yp,R2~Yp).Scaff(mapk!1) -> MAPK(s,R1~Yp,R2~Yp) + Scaff(mapk)
MAPK(s!1,R1~Yp,R2~Yp),Scaff(mapk!1) -> MAPK(s,R1~Yp,R2~Yp) , Scaff(mapk)

#MAP3K(s!1,S~A).Scaff(map3k!1,map2k!2).MAP2K(s!2,R1~Y) -> MAP3K(s!1,S~A).Scaff(map3k!1,map2k!2).MAP2K(s!2,R1~Yp)
MAP3K(s!1,S~A),Scaff(map3k!1,map2k!2),MAP2K(s!2,R1~Y) -> MAP3K(s!1,S~A),Scaff(map3k!1,map2k!2),MAP2K(s!2,R1~Yp)

#MAP3K(s!1,S~A).Scaff(map3k!1,map2k!2).MAP2K(s!2,R2~Y) -> MAP3K(s!1,S~A).Scaff(map3k!1,map2k!2).MAP2K(s!2,R2~Yp)
MAP3K(s!1,S~A),Scaff(map3k!1,map2k!2),MAP2K(s!2,R2~Y) -> MAP3K(s!1,S~A),Scaff(map3k!1,map2k!2),MAP2K(s!2,R2~Yp)

#MAPK(s!1,R1~Y).Scaff(mapk!1,map2k!2).MAP2K(s!2,R1~Yp,R2~Yp) -> MAPK(s!1,R1~Yp).Scaff(mapk!1,map2k!2).MAP2K(s!2,R1~Yp,R2~Yp)
MAPK(s!1,R1~Y),Scaff(mapk!1,map2k!2),MAP2K(s!2,R1~Yp,R2~Yp) -> MAPK(s!1,R1~Yp),Scaff(mapk!1,map2k!2),MAP2K(s!2,R1~Yp,R2~Yp)

#MAPK(s!1,R2~Y).Scaff(mapk!1,map2k!2).MAP2K(s!2,R1~Yp,R2~Yp) -> MAPK(s!1,R2~Yp).Scaff(mapk!1,map2k!2).MAP2K(s!2,R1~Yp,R2~Yp)
MAPK(s!1,R2~Y),Scaff(mapk!1,map2k!2),MAP2K(s!2,R1~Yp,R2~Yp) -> MAPK(s!1,R2~Yp),Scaff(mapk!1,map2k!2),MAP2K(s!2,R1~Yp,R2~Yp)

#MAP3K(S~A) -> MAP3K(S~I)
MAP3K(S~A) -> MAP3K(S~I)

#MAP2K(R1~Yp) -> MAP2K(R1~Y)
MAP2K(R1~Yp) -> MAP2K(R1~Y)

#MAP2K(R2~Yp) -> MAP2K(R2~Y)
MAP2K(R2~Yp) -> MAP2K(R2~Y)

#MAPK(R1~Yp) -> MAPK(R1~Y)
MAPK(R1~Yp) -> MAPK(R1~Y)

#MAPK(R2~Yp) -> MAPK(R2~Y)
MAPK(R2~Yp) -> MAPK(R2~Y)
\end{lstlisting}

\end{document}